\documentclass{elsarticle}

\usepackage{lineno,hyperref}
\usepackage{amsfonts}
\usepackage{amsmath}
\usepackage{amsthm} 
\usepackage{enumerate}
\usepackage{makecell}
\usepackage{color}
\usepackage[boxruled]{algorithm2e} 
\newtheorem{rrule}{R-Rule}
\newtheorem{definition}{Definition}
\newtheorem{prop}{Proposition}
\newtheorem{lemma}{Lemma}

\newtheorem{theorem}{Theorem}

\modulolinenumbers[5]

\begin{document}

\begin{frontmatter}

    \title{Further Improvements for SAT in Terms of Formula Length}

    \author{Junqiang Peng}
    \ead{jqpeng0@foxmail.com}
    \author{Mingyu Xiao\corref{correspondingauthor}}
    \cortext[correspondingauthor]{Corresponding author}
    \ead{myxiao@gmail.com}
    \address{School of Computer Science and Engineering, \\
        University of Electronic Science and Technology of China, China}

    \begin{abstract}
        In this paper, we prove that the general CNF satisfiability problem can be solved in $O^*(1.0638^L)$ time, where $L$ is the length of the input CNF-formula (i.e., the total number of literals in the formula), which improves the previous result of $O^*(1.0652^L)$ obtained in 2009.
        Our algorithm was analyzed by using the measure-and-conquer method.
        Our improvements are mainly attributed to the following two points: we carefully design branching rules to deal with degree-5 and degree-4 variables to avoid previous bottlenecks; we show that some worst cases will not always happen, and then we can use an amortized technique to get further improvements.
        In our analyses, we provide some general frameworks for analysis and several lower bounds on the decreasing of the measure to simplify the arguments.
        These techniques may be used to analyze more algorithms based on the measure-and-conquer method.
    \end{abstract}

    \begin{keyword}
        \texttt{Satisfiability} \sep
        \texttt{Parameterized Algorithms} \sep
        \texttt{Measure-and-Conquer} \sep
        \texttt{Amortized Analysis} \sep
        \texttt{Worst-Case Analysis}
    \end{keyword}

\end{frontmatter}

\section{Introduction}
Propositional Satisfiability is the problem of determining, for a formula of the propositional calculus, if there is an assignment of truth values to its variables for which that formula evaluates to true.
By SAT, we mean the problem of propositional satisfiability for formulas in conjunctive normal form (CNF)~\cite{DBLP:conf/dimacs/CookM96}.
SAT is the first problem proved to be NP-complete~\cite{DBLP:conf/stoc/Cook71}, and it plays an important role in computational complexity, artificial intelligence, and many others~\cite{DBLP:series/faia/2009-185}.
There are numerous investigations on this problem in different fields, such as approximation algorithms, randomized algorithms, heuristic algorithms, and exact and parameterized algorithms.
In this paper, we study parameterized algorithms for SAT parameterized by the input length.

In order to measure the running time bound for SAT, there are three frequently used parameters: the number of variables $n$, the number of clauses $m$, and the input length $L$. The input length $L$ is defined as the sum of the number of literals in each clause.
The number of variables $n$ should be the most basic parameter. The simple brute force algorithm to try all $2^n$ possible assignments of the $n$ variables will get the running time bound of $O^*(2^n)$.\footnote{The $O^*$ notation supervises all polynomial factors, i.e., $f(n)=O^*(g(n))$ means $f(n)=O(g(n)n^{O(1)})$.}
After decades of hard work, no one can break this trivial bound.
The Strong Exponential Time Hypothesis (SETH) conjectures that SAT cannot be solved in time $O^*(c^n)$ for some constant $c<2$~\cite{DBLP:journals/jcss/ImpagliazzoP01}.
For a restricted version, the $k$-SAT problem (the length of each clause in the formula is bounded by a constant $k$) can be solved in $O^*(c(k)^n)$ time for some value $c(k)<2$ depending on $k$. There are two major research lines: one is deterministic local search algorithms like Sch{\"{o}}ning's algorithm~\cite{DBLP:conf/focs/Schoning99,DBLP:conf/icalp/Liu18} and the other one is randomized algorithms based on PPZ~\cite{DBLP:journals/cjtcs/PaturiPZ99} or PPSZ~\cite{DBLP:journals/jacm/PaturiPSZ05, DBLP:conf/focs/Scheder21}.
For example, 3-SAT can be deterministically solved in $O^*(1.3279^n)$ time~\cite{DBLP:conf/icalp/Liu18} and 4-SAT can be deterministically solved in $O^*(1.4986^n)$ time~\cite{DBLP:conf/icalp/Liu18}.

When it comes to the parameter $m$, Monien \textit{et al.}~\cite{monien1981upper} first gave an algorithm with time complexity $O^*(1.260^m)$ in 1981. Later, the bound was improved to $O^*(1.239^m)$ by Hirsch~\cite{DBLP:conf/soda/Hirsch98} in 1998, and then improved to $O^*(1.234^m)$ by Yamamoto~\cite{DBLP:conf/isaac/Yamamoto05} in 2005. Now the best result is $O^*(1.2226^m)$ obtained by Chu, Xiao, and Zhang~\cite{DBLP:journals/tcs/ChuXZ21}.

The input length $L$ is another important and frequently studied parameter. It is probably the most precise parameter to describe the size of the input CNF-Formula.
From the first algorithm with running time bound $O^*(1.0927^L)$ by Van Gelder~\cite{DBLP:journals/iandc/Gelder88} in 1988, the result was improved several times. In 1997, the bound was improved to $O^*(1.0801^L)$ by Kullmann and Luckhardt~\cite{kullmann1997deciding}. Later, the bound was improved to $O^*(1.0758^L)$ by Hirsch~\cite{DBLP:conf/soda/Hirsch98} in 1998, and improved again by Hirsch~\cite{DBLP:journals/jar/Hirsch00a} to $O^*(1.074^L)$ in 2000.
Then Wahlstr{\"{o}}m~\cite{DBLP:conf/esa/Wahlstrom05} gave an $O^*(1.0663^L)$-time algorithm in 2005.
In 2009, Chen and Liu~\cite{DBLP:conf/wads/ChenL09} achieved a bound $O^*(1.0652^L)$ by using the measure-and-conquer method.
In this paper, we further improve the result to $O^*(1.0638^L)$.
We list the major progress and our result in Table \ref{tb-result-L}.

\begin{table}[t]\label{tb-result-L}
\begin{center}
\caption{Previous and our upper bound for SAT}
\begin{tabular}{l|c}
    \hline
    Running time bounds & References \\
    \hline
    $O^*(1.0927^L)$ & Van Gelder~1988~\cite{DBLP:journals/iandc/Gelder88} \\
    $O^*(1.0801^L)$ & Kullmann and Luckhardt~1997~\cite{kullmann1997deciding} \\
    $O^*(1.0758^L)$ & Hirsch~1998~\cite{DBLP:conf/soda/Hirsch98}\\
    $O^*(1.074^L)$ & Hirsch~2000~\cite{DBLP:journals/jar/Hirsch00a} \\
    $O^*(1.0663^L)$ & Wahlstr{\"{o}}m~2005~\cite{DBLP:conf/esa/Wahlstrom05} \\
    $O^*(1.0652^L)$ & Chen and Liu~2009~\cite{DBLP:conf/wads/ChenL09}\\
    $O^*(1.0638^L)$ & \textbf{This paper} \\
    \hline
\end{tabular}
\end{center}
\end{table}

It is also worth mentioning the maximum satisfiability problem (MaxSAT), which is strongly related to the SAT problem.
SAT asks whether we can give an assignment of the variables to satisfy all clauses, while MaxSAT asks us to satisfy the maximum number of clauses.
For MaxSAT, researchers usually consider the decision version of it: whether we can satisfy at least $k$ clauses.
Thus, $k$ is a natural parameter in parameterized algorithms. The current best result for MaxSAT parameterized by $k$ is $O^*(1.325^k)$~\cite{DBLP:journals/tcs/ChenXW17}.
The three parameters $n$, $m$, and $L$ mentioned above are also frequently considered for MaxSAT.
For the number of variables $n$, similar to SAT, the bound $O^*(2^n)$ obtained by a trivial algorithm is still not broken so far, and it is impossible to break under SETH.
In terms of the number of clauses $m$, very recently, the running time bound of MaxSAT was improved to $O^*(1.2886^m)$~\cite{DBLP:conf/ijcai/Xiao22}.
When it comes to the length of the formula $L$, the bound was also recently improved to $O^*(1.0927^L)$~\cite{DBLP:conf/aaai/AlferovB21}.

Our algorithm, as well as most algorithms for SAT-related problems, is based on the branch-and-search paradigm.
The idea of the branch-and-search algorithm is simple and practical: for a given CNF-formula $\mathcal{F}$, we iteratively branch on a variable or literal $x$ into two branches by assigning value 1 or 0 to it. Let $\mathcal{F}_{x=1}$ and $\mathcal{F}_{\overline{x}=1}$ be the resulting CNF-formulas by assigning value 1 and 0 to $x$, respectively. It holds that $\mathcal{F}$ is satisfiable if and only if at least one of $\mathcal{F}_{x=1}$ and $\mathcal{F}_{\overline{x}=1}$ is satisfiable. To get a running time bound, we need to analyze how much the parameter $L$ can decrease in each branch. To break bottlenecks in direct analysis, some references~\cite{DBLP:conf/esa/Wahlstrom05,DBLP:conf/wads/ChenL09} analyzed
the algorithms based on new measures and gave the relation between the new measures and $L$.
The measure-and-conquer method is one of the frequently used techniques.
Our algorithm in this paper will also adopt the measure-and-conquer method and deal with variables from high degree to low degree.
We compare our algorithm with the previous algorithm~\cite{DBLP:conf/wads/ChenL09} based on the measure-and-conquer method.
The algorithm in~\cite{DBLP:conf/wads/ChenL09} carefully analyzed branching operations for variables of degree $4$ and used a simple and uniform rule to deal with variables of degree at least 5. Their bottlenecks are some cases to deal with degree-4 variables. Our algorithm will carefully analyze the branching operation for degree-$5$ variables and use simple rules for degree-$4$ variables. In the conference version of this paper~\cite{DBLP:conf/sat/PengX21}, we have shown that these modifications led to improvements.
In this version, we further show that the bottlenecks will not always happen, and we can use an amortized technique to combine bottleneck cases with the following good cases to get a further improved average bound.
To do this, we also need to modify some steps of the algorithm in the conference version~\cite{DBLP:conf/sat/PengX21} and carefully design and analyze the branching operations for some special structures to satisfy the requirement of amortization. Finally, we are able to improve the result to $O^*(1.0638^L)$.
In our algorithm, to simplify some case analyses and arguments, we provide some general frameworks for analysis and establish several lower bounds on the decreasing of the measure. The lower bounds may reveal some structural properties of the problem, and the analysis framework may be used to analyze more algorithms based on the measure-and-conquer method.

\section{Preliminaries}

Let $V = \{x_1, x_2, ..., x_n\}$ denote a set of $n$ boolean \emph{variables}.
Each variable $x_i$ ($i \in \{1, 2, ..., n\}$) has two corresponding \textit{literals}: the positive literal $x_i$ and the negative literal $\overline{x_i}$ (we use $\overline{x}$ to denote the negation of a literal $x$, and $\overline{\overline{x}}=x$).
A \textit{clause} on $V$ consists of some literals on $V$.
Note that we allow a clause to be empty.
A clause containing literal $z_1,z_2,\dots, z_q$ is simply written as $z_1z_2\dots z_q$. Thus, we use $zC$ to denote the clause containing literal $z$ and all literals in clause $C$. We also use $C_1C_2$ to denote the clause containing all literals in clauses $C_1$ and $C_2$.
We use $\overline{C}$ to denote a clause that contains the negation of every literal in clause $C$. That is, if $C=z_1z_2...z_q$, then $\overline{C} = \overline{z_1}\overline{z_2}...\overline{z_q}$.
A \emph{CNF-formula} on $V$ is the conjunction of a set of clauses $\mathcal{F} = C_1\wedge C_2\wedge  ...\wedge C_m$. When we say a variable $x$ is contained in a clause (resp., a formula), it means that the clause (resp., at least one clause of the formula) contains the literal $x$ or its negation $\overline{x}$.

An assignment for $V$ is a map $A: V\rightarrow \{0, 1\}$. A clause $C_j$ is \textit{satisfied} by an assignment if and only if there exists at least one literal in $C_j$ such that the assignment makes its value 1. A CNF-formula is \emph{satisfied} by an assignment $A$ if and only if each clause in it is satisfied by $A$. We say a CNF-formula is \textit{satisfiable} if it can be satisfied by at least one assignment.
We may assign value $0$ or $1$ to a literal, which is indeed to assign a value to its variable to make the corresponding literal 0 or 1.

A literal $z$ is called an \textit{$(i, j)$-literal} (resp., an \textit{$(i^+, j)$-literal} or \textit{$(i^-, j)$-literal}) in a formula $\mathcal{F}$ if $z$ appears $i$ (resp., at least $i$ or at most $i$) times and $\overline{z}$ appears $j$ times in the formula $\mathcal{F}$. Similarly, we can define $(i, j^+)$-literal, $(i, j^-)$-literal, $(i^+, j^+)$-literal, $(i^-, j^-)$-literal, and so on.
Note that literal $z$ is an $(i, j)$-literal if and only if literal $\overline{z}$ is a $(j, i)$-literal.
A variable $x$ is an \textit{$(i, j)$-variable} if the positive literal $x$ is an $(i, j)$-literal.

For a variable or a literal $x$ in formula $\mathcal{F}$, the \textit{degree} of it, denoted by $deg(x)$, is the number of $x$ appearing in $\mathcal{F}$ plus the number of $\overline{x}$ appearing in $\mathcal{F}$, i.e., $deg(x)=i+j$ for an $(i, j)$-variable or $(i, j)$-literal $x$.
A \textit{$d$-variable} (resp., \emph{$d^+$-variable} or \emph{$d^-$-variable}) is a variable with the degree exactly $d$ (resp., at least $d$ or at most $d$). The degree of a formula $\mathcal{F}$ is the maximum degree of all variables in $\mathcal{F}$. For a clause or a formula $C$, the set of variables whose literals appear in $C$ is denoted by $var(C)$.

The $\textit{length}$ of a clause $C$, denoted by $|C|$, is the number of literals in $C$. A clause is a \textit{$k$-clause} or \textit{$k^+$-clause} if the length of it is $k$ or at least $k$. We use $L(\mathcal{F})$ to indicate the length of a formula $\mathcal{F}$. It is the sum of the lengths of all clauses in $\mathcal{F}$, which is also the sum of the degrees of all variables in $\mathcal{F}$. A formula $\mathcal{F}$ is called \textit{$k$-CNF formula} if each clause in $\mathcal{F}$ has a length of at most $k$.

In a formula $\mathcal{F}$, a literal $x$ is called a \emph{neighbor} of a literal $z$ if there is a clause containing both $z$ and $x$. The set of neighbors of a literal $z$ in a formula $\mathcal{F}$ is denoted by \textit{$N(z, \mathcal{F})$}. We also use $N^{(k)}(x, \mathcal{F})$ (resp., $N^{(k+)}(z, \mathcal{F})$) to denote the neighbors of $z$ in $k$-clauses (resp., $k^+$-clauses) in $\mathcal{F}$, i.e., for any $z'\in N^{(k)}(z, \mathcal{F})$ (resp., $z'\in N^{(k+)}(z, \mathcal{F})$), there exists a $k$-clause (resp., $k^+$-clause) containing both $z$ and $z'$. If we say a variable $y$ appears in $N(x, \mathcal{F})$, it means that literal $y$ or literal $\overline{y}$ is in $N(x, \mathcal{F})$.

\section{Some Techniques}

\subsection{Branch-and-Search Algorithms}
Our algorithm is a standard branch-and-search algorithm, which first applies some reduction rules to reduce the instance as much as possible and then searches for a solution by branching. The branching operations may exponentially increase the running time.
We will use a measure to evaluate the size of the search tree generated in the algorithm.
For the SAT problem, the number of variables or clauses of the formula is a commonly used measure.
More fundamentals of branching heuristics about the SAT problem can be found in~\cite{DBLP:series/faia/2009-185}.

Let $\mu$ be a measure of an instance. We use $T(\mu)$ to denote the number of leaves of the search tree generated by the algorithm for any instance with the measure being at most $\mu$. For a branching operation that branches on the current instance into $l$ branches with the measure decreasing by at least $a_i$ in the $i$-th branch, we get a recurrence relation
\[T(\mu)\leq T(\mu-a_1)+T(\mu-a_2)+\dots + T(\mu-a_l).\]
The recurrence relation can also be simply represented by a \emph{branching vector} $[a_1,a_2,\dots,a_l]$.
The largest root of the function $f(x) = 1-\sum_{i=1}^lx^{-a_i}$, denoted by $\tau(a_1, a_2, \dots, a_l)$, is called the \emph{branching factor} of the recurrence.
If the maximum branching factor for all branching operations in the algorithm is at most $\gamma$, then $T(\mu)=O(\gamma^\mu)$.
If the algorithm runs in polynomial time on each node of the search tree, then the total running time of the algorithm is $O^*(\gamma^\mu)$.
More details about analyzing recurrences can be found in the monograph~\cite{DBLP:series/txtcs/FominK10}.

In the analysis, we need to find the largest branching factor in the algorithm.
Let $\mathbf{a}=[a_1,a_2,\dots,a_l]$ and $\mathbf{b}=[b_1,b_2,\dots,b_l]$ be two branching vectors.
If $\tau(a_1,a_2,\dots,a_l)>\tau(b_1,b_2,\dots,b_l)$, then we say $\mathbf{a}$ \emph{covers} $\mathbf{b}$.
Usually, it is hard to compare two branching vectors directly. In this paper, we will frequently use some special cases.
If it holds that $a_i\leq b_i$ for all $i=1,2\dots,l$, then the branching factor of $\mathbf{b}$ is not smaller than that of $\mathbf{a}$, i.e., $\mathbf{a}$ covers $\mathbf{b}$.
Let $i$ and $j$ be positive reals such that $0<i<j$, for any $0<\epsilon<\frac{j-i}{2}$, it holds that $\tau(i,j)>\tau(i+\epsilon, j-\epsilon)$.
For this case, we have that the branching vector $[i,j]$ covers $[i+\epsilon, j-\epsilon]$.

\subsection{Shift}

In some cases, the worst branch in the algorithm will not always happen.
In order to deal with this situation, one can use an amortization technique to get improved results.
This technique has been used in several previous papers, see~\cite{DBLP:conf/iwpec/Iwata11, DBLP:journals/algorithmica/ChenKX05, DBLP:books/daglib/Gasper2010,DBLP:journals/iandc/XiaoN17}.
In this paper, we will follow the notation ``shift'' in~\cite{DBLP:journals/iandc/XiaoN17} to implement the amortized analysis.

Consider two branching operations $A$ and $B$ with branching vectors $\mathbf{a}=[a_1, a_2]$ and $\mathbf{b}=[b_1, b_2]$ (with recurrences $T(\mu)\leq T(\mu-a_1) + T(\mu-a_2)$ and $T(\mu)\leq T(\mu-b_1) + T(\mu-b_2)$) such that the branching operation $B$ has a smaller branching factor than $A$ does, where $\mathbf{a}$ may be the bottleneck in the running time analysis of the algorithm ($\mathbf{a}$ has the maximum branching factor among all branching vectors).
Suppose branching operation $B$ is always applicable to the sub-instance $\mathcal{F}_1$ generated by the first sub-branch of $A$.
In this case, we can get a better branching vector $\mathbf{c} = [a_1+b_1,a_1+b_2,a_2]$ by combining the branching operation $A$ and the branching operation $B$ applied to $\mathcal{F}_1$.
Note that the branching operation $B$ may also be applicable to the sub-instances generated by several branching operations other than $A$.
In order to ease such an analysis without generating all combined branching vectors, we use a notion of ``shift"~\cite{DBLP:journals/iandc/XiaoN17}.
We transfer some amount from the measure decreases in the recurrence for $B$ to that for $A$ as follows.
We save an amount $\sigma > 0$ of measure decreases from $B$ by evaluating the branching operation $B$ with branching vectors
\[
    [b_1-\sigma, b_2-\sigma],
\]
which leads to a larger branching factor than its original branching vector. The saved measure decrease $\sigma$ will be included in the branching vectors for branching operation $A$ to obtain
\[
    [a_1+\sigma, a_2].
\]
The saved amount $\sigma$ is also called a \textit{shift}, where the best value for $\sigma$ will be determined so that the maximum branching factor is minimized. Clearly, we can get branching vector $[a_1+b_1,a_1+b_2,a_2]$ by combining the above two branching vectors.
In our analysis, we will use one shift $\sigma$.

\subsection{Measure and Conquer}
The measure-and-conquer method \cite{DBLP:journals/jacm/FominGK09} is a powerful tool for analyzing the branch-and-search algorithms.
The main idea of the method is to adopt a new measure in the analysis of the algorithm. For example, instead of using the number of variables as the measure, it may set weights to different variables and use the sum of all variable weights as the measure. This method may be able to catch more structural properties and then get further improvements. Nowadays, the fastest exact algorithms for many NP-hard problems were designed by using this method~\cite{DBLP:journals/iandc/XiaoN17, DBLP:journals/algorithmica/XiaoN16, DBLP:journals/dam/RooijB11}. In this paper, we will also use the measure-and-conquer method.

We introduce a weight to each variable in the formula according to the degree of the variable, $w\colon \mathbb{Z}^+ \rightarrow \mathbb{R}^+$, where $\mathbb{Z}^+$ and $\mathbb{R}^+$ denote the sets of nonnegative integers and nonnegative reals, respectively. Let $w_i$ denote the weight of a variable with degree $i$.
A variable with a lower degree will not receive a higher weight. i.e., $w_i\geq w_{i-1}$.
In our algorithm, the measure of a formula $\mathcal{F}$ is defined as
\begin{equation}\label{measure}
    \mu(\mathcal{F}) = \sum_{x} w_{deg(x)}.
\end{equation}
In other words, $\mu(\mathcal{F})$ is the sum of the weight of all variables in $\mathcal{F}$.
Let $n_i$ denote the number of $i$-variables in $\mathcal{F}$. Then we also have that
\[
    \mu(\mathcal{F})=\sum_{i} w_in_i.
\]

One important step is to set the value of weight $w_i$. Different values of $w_i$ will generate different branching vectors and factors. We need to find a good setting of $w_i$ so that the worst branching factor is as small as possible.
We will get the value of $w_i$ by solving a quasiconvex program after listing all our branching vectors. However, we pre-specify some requirements of the weights to simplify arguments. Some similar assumptions were used in the previous measure-and-conquer analysis. We set the weight such that
\begin{equation}\label{weight1}
    \begin{aligned}
         & w_1=w_2=0,                       \\
         & 0<w_3< 2, w_4= 2w_3, ~\text{and} \\
         & w_i = i  ~\text{for }~ i\geq 5.
    \end{aligned}
\end{equation}
We use $\delta_i$ to denote the difference between $w_i$ and $w_{i-1}$ for $i > 0$, i.e., $\delta_i=w_i-w_{i-1}$.
By (\ref{weight1}), we have \begin{equation}\label{eq-w3-d3-d4}
    w_3=\delta_3=\delta_4,
\end{equation}
and $2w_5> 5w_3\Rightarrow 2w_5-4w_3>w_3\Rightarrow 2w_5-2w_4>w_3$, which implies
\begin{equation}\label{ineq-2d5>w3}
    2\delta_5>w_3.
\end{equation}

We also assume that
\begin{equation}\label{weight2}
    \begin{aligned}
         & w_3 \geq \delta_5,  \\
         & 1\leq \delta_i \leq \delta_{i-1} ~\text{for}~i\geq 3, ~\text{and}\\
         & w_3-\delta_5 < 1.
    \end{aligned}
\end{equation}
Under these assumptions, it holds that $w_i\leq i$ for each $i$. Thus, for any formula $\mathcal{F}$, it always holds that
\begin{equation}\label{measure-leq-L}
    \mu(\mathcal{F}) \leq L(\mathcal{F}).
\end{equation}
This tells us that if we can get a running time bound of $O^*(c^{\mu(\mathcal{F})})$ for a real number $c$, then we also get a running time bound of $O^*(c^{L(\mathcal{F})})$ for this problem.
To obtain a running time bound in terms of the formula length $L(\mathcal{F})$, we consider the measure $\mu(\mathcal{F})$ and show how much the measure $\mu(\mathcal{F})$ decreases in the branching operations of our algorithm and find the worst branching factor among all branching vectors.

\section{The Algorithm}

We will first introduce our algorithm and then analyze its running time bound by using the measure-and-conquer method.
Our algorithm consists of reduction operations and branching operations. When no reduction operations can be applied anymore, the algorithm will search for a solution by branching. We first introduce our reduction rules.
\subsection{Reduction Rules}
We have ten reduction rules. Most are well-known and frequently used in the literature (see ~\cite{DBLP:conf/esa/Wahlstrom05,DBLP:conf/wads/ChenL09} for examples).
We introduce the reduction rules in the order as stated, and a reduction rule will be applied in our algorithm only when all the previous reduction rules do not apply to the instance.

\begin{rrule}[Elimination of duplicated literals]\label{rule-dup}
    If a clause $C$ contains duplicated literals $z$, remove all but one $z$ in $C$.
\end{rrule}

\begin{rrule}[Elimination of subsumptions]\label{rule-subsump}
    If there are two clauses $C$ and $D$ such that $C\subseteq D$, remove clause $D$.
\end{rrule}

\begin{rrule}[Elimination of tautology]\label{rule-tauto}
    If a clause $C$ contains two opposite literals $z$ and $\overline{z}$, remove clause $C$.
\end{rrule}

\begin{rrule}[Elimination of $1$-clauses and pure literals]\label{rule-pure}
    If there is a $1$-clause $\{x\}$ or a $(1^+, 0)$-literal $x$, assign $x=1$.
\end{rrule}

\textit{Davis-Putnam Resolution}, proposed in~\cite{DBLP:journals/jacm/DavisP60}, is a classic and frequently used technology for SAT.
Let $\mathcal{F}$ be a CNF-formula and $x$ be a variable in $\mathcal{F}$.
Assume that clauses containing literal $x$ are $xC_1,xC_2,...,xC_a$ and clauses containing literal $\overline{x}$ are $\overline{x}D_1, \overline{x}D_2, ..., \overline{x}D_b$. A \emph{Davis-Putnam resolution} on $x$ is to construct a new CNF-formula $DP_{x}(\mathcal{F})$ by the following method: initially $DP_{x}(\mathcal{F}) = \mathcal{F}$;
add new clauses $C_iD_j$ for each $1\leq i\leq a$ and  $1\leq j\leq b$; and remove $xC_1,xC_2,...,xC_a, \overline{x}D_1, \overline{x}D_2,...,\overline{x}D_b$ from the formula. It is known that
\begin{prop}[\cite{DBLP:journals/jacm/DavisP60}]
    A CNF-formula $\mathcal{F}$ is satisfiable if and only if $DP_{x}(\mathcal{F})$ is satisfiable.
\end{prop}
In the resolution operation, each new clause $C_iD_j$ is called a \emph{resolvent}. A resolvent is \emph{trivial} if it contains both
a literal and the negation of it. Since trivial resolvents will always be satisfied, we can simply delete trivial resolvents from the instance directly. So when we do resolutions, we assume that all trivial resolvents will be deleted.

\begin{rrule}[Trivial resolution]\label{rule-trivial-resol}
    If there is a variable $x$ such that the degree of each variable in $DP_{x}(\mathcal{F})$ is not greater than that in $\mathcal{F}$, then apply resolution on $x$.
\end{rrule}

\begin{rrule}[\cite{DBLP:conf/wads/ChenL09}]\label{rule-2cls-1neg}
    If there are a $2$-clause $z_1z_2$ and a clause $C$ containing both $z_1$ and $\overline{z_2}$,
    then  remove $\overline{z_2}$ from $C$.
\end{rrule}

\begin{rrule}\label{rule-1liter-neg}
    If there are two clauses $z_1z_2C_1$ and $z_1\overline{z_2}C_2$, where literal $\overline{z_2}$ appears in no other clauses, then remove $z_1$ from clause $z_1z_2C_1$.
\end{rrule}
\begin{lemma}\label{proof-Rule7}
    Let $\mathcal{F}$ be a CNF-formula and $\mathcal{F}'$ be the resulting formula after applying R-Rule~\ref{rule-1liter-neg} on $\mathcal{F}$. Then $\mathcal{F}$ is satisfiable if and only if $\mathcal{F}'$ is satisfiable.
\end{lemma}
\begin{proof}
    Let $\mathcal{F}$ be the original formula and $\mathcal{F}'$ be the formula after replacing clause $z_1z_2C_1$ with clause $z_2C_1$
    in $\mathcal{F}$. Clearly, if $\mathcal{F}'$ is satisfied, then $\mathcal{F}$ is satisfied. We consider the other direction.

    Assume that $\mathcal{F}$ is satisfied by an assignment $A$. We show that $\mathcal{F}'$ is also satisfied.
    If  $z_2C_1$ is satisfied by assignment $A$, then $\mathcal{F}'$ is satisfied by assignment $A$.
    Next, we assume assignment $A$ satisfies $z_1z_2C_1$ but not $z_2C_1$. Then in $A$, we have that $z_1=1$ and $z_2=0$.
    Since $\overline{z_2}$ is a $(1, 1^+)$-literal and $z_1\overline{z_2}C_2$ is the only clause containing $\overline{z_2}$, we know that all clauses will be satisfied if we replace $z_2=0$ with $z_2=1$ in $A$. Thus, $\mathcal{F}'$ is still satisfied.
\end{proof}

\begin{rrule}[\cite{DBLP:conf/wads/ChenL09}]\label{rule-2cls-2neg}
    If there is a $2$-clause $z_1z_2$ and a clause $\overline{z_1}\overline{z_2}C$ such that literal $\overline{z_1}$ appears in no other clauses, remove the clause $z_1z_2$ from $\mathcal{F}$.
\end{rrule}

\begin{rrule}[\cite{DBLP:conf/wads/ChenL09}]\label{rule-link}
    If there is a $2$-clause $z_1z_2$ such that either literal $z_1$ appears only in this clause or there is another $2$-clause $\overline{z_1}\overline{z_2}$, then replace $z_1$ with $\overline{z_2}$ in $\mathcal{F}$ and then apply R-Rule \ref{rule-tauto} as often as possible.
\end{rrule}

\begin{rrule}[\cite{DBLP:conf/wads/ChenL09}]\label{rule-back-resol}
    If there are two clauses $CD_1$ and $CD_2$ such that $|D_1|, |D_2|\geq 1$ and $|C|\geq 2$, then remove $CD_1$ and $CD_2$ from $\mathcal{F}$, and add three new
    clauses $xC$, $\overline{x}D_1$, and $\overline{x}D_2$, where $x$ is a new $3$-variable.
\end{rrule}

R-Rule~\ref{rule-back-resol} is like the Davis-Putnam resolution in reverse, and thus it is correct.

\begin{definition}[Reduced formulas]\label{def-reducedf}
    A CNF-formula $\mathcal{F}$ is called \textit{reduced}, if none of the above reduction rules can be applied on it.
\end{definition}

Our algorithm will first iteratively apply the above reduction rules to get a reduced formula.
We will use $R(\mathcal{F})$ to denote the resulting reduced formula obtained from $\mathcal{F}$.
Next, we show some properties of reduced formulas.

\begin{lemma}\label{prop1}
    In a reduced CNF-formula $\mathcal{F}$, all variables are $3^+$-variables.
\end{lemma}
\begin{proof}
    If there is a $1$-variable, then R-Rule \ref{rule-pure} could be applied. For a $2$-variable in $\mathcal{F}$, if it is a $(2, 0)$-variable or $(0, 2)$-variable, then R-Rule \ref{rule-pure} could be applied; if it is a $(1, 1)$-variable, R-Rule \ref{rule-trivial-resol} could be applied.
\end{proof}

\begin{lemma}\label{prop2}
    In a reduced CNF-formula $\mathcal{F}$, if there is a $2$-clause $xy$, then no other clause in $\mathcal{F}$ contains $xy$, $\overline{x}y$, or $x\overline{y}$.
\end{lemma}
\begin{proof}
    If there is a clause containing $xy$, then R-Rule \ref{rule-subsump} could be applied. If there is a clause containing $\overline{x}y$ or $x\overline{y}$, then R-Rule \ref{rule-2cls-1neg} could be applied.
\end{proof}

\begin{lemma}\label{prop4}
    In a reduced CNF-formula $\mathcal{F}$, if there is a clause $xyC$, then
    \begin{enumerate}[(i)]
        \item no other clause contains $xy$; \label{prop4-1}
        \item no other clause contains $\overline{x}y$ or $\overline{x}\overline{y}$ if $x$ is a $3$-variable. \label{prop4-2}
    \end{enumerate}
\end{lemma}
\begin{proof}
    (\ref{prop4-1}): If there is a clause containing $xy$, then either R-Rule~\ref{rule-subsump} or R-Rule \ref{rule-back-resol} can be applied.

    (\ref{prop4-2}): Since $x$ is a $3$-variable and all $(1^+, 0)$-literals are reduced by R-Rule \ref{rule-pure}, $x$ is actually a $(1, 2)$-variable or $(2, 1)$-variable. If there is a clause containing $\overline{x}y$, R-Rule \ref{rule-1liter-neg} would be applicable. If there is a clause containing $\overline{x}\overline{y}$, then there is only one non-trivial resolvent after resolving on $x$, which means R-Rule \ref{rule-trivial-resol} would be applicable.
\end{proof}

\begin{lemma}\label{prop5}
    In a reduced CNF-formula $\mathcal{F}$, if there is $(1, 2^+)$-literal $x$ (let $xC$ be the only clause containing $x$), then
    \begin{enumerate}[(i)]
        \item $|C|\geq 2$; \label{prop5-1}
        \item $var(C)\cap var(N^{(2)}(\overline{x}, F)) = \emptyset$, that is, if $y\in N^{(2)}(\overline{x}, F)$, then $y, \overline{y}\notin C$. \label{prop5-2}
    \end{enumerate}
\end{lemma}
\begin{proof}
    (\ref{prop5-1}): If $|C|=0$, then R-Rule \ref{rule-pure} could be applied; if $|C|=1$, which means $xC$ is a $2$-clause, then R-Rule \ref{rule-link} could be applied since $x$ is a $(1, 2^+)$-literal.

    (\ref{prop5-2}): Assume that $y$ is a literal in $N^{(2)}(\overline{x}, F)$, in other words, there is a $2$-clause $\overline{x}y$. Note that $x$ is a $(1, 2^+)$-literal. If $y\in C$, then R-Rule \ref{rule-1liter-neg}
    could be applied; if $\overline{y}\in C$, then R-Rule \ref{rule-2cls-2neg} could be applied.
\end{proof}

\subsection{Branching Rules and the Algorithm}
After getting a reduced formula, we will search for a solution by branching.
In a branching operation, we will generate two smaller CNF-formulas such that the original formula is satisfiable if and only if at least one of the two new formulas is satisfiable.
The two smaller formulas are generated by specifying the value of a set of literals in the original formula.

The simplest branching rule is that we pick up a variable or literal $x$ from $\mathcal{F}$ and branch into two branches $\mathcal{F}_{x=1}$ and $\mathcal{F}_{x=0}$, where $\mathcal{F}_{x=1}$ and $\mathcal{F}_{x=0}$ are the formulas after assigning $x=1$ and $x=0$ in $\mathcal{F}$, respectively.
When the picked literal $x$ is a $(1,1^+)$-literal, we will apply a stronger branching.
Assume that $xC$ is the only clause containing $x$. Then we branch into two branches $\mathcal{F}_{x=1\And C=0}$ and $\mathcal{F}_{x=0}$,
where $\mathcal{F}_{x=1\And C=0}$ is the resulting formula after assigning 1 to $x$ and $0$ to all literals in $C$ in
$\mathcal{F}$.
The correctness of this branching operation is also easy to observe. Only when all literals in $C$ are assigned $0$ do we need to assign $1$ to $x$.
Generally, we always pick a variable or literal with the maximum degree in the formula to branch.

The main steps of our algorithm are given in Algorithm~\ref{main-algorithm}. The algorithm will execute one step only when all previous steps can not be applied.
In Step~2, the algorithm first reduces the formula by applying the reduction rules.
Afterwards, in Steps~3-15, the algorithm deals with variables from high-degree to low-degree by branching.
In Step~3, the algorithm branches on variables with degree $\geq 6$.
Steps~4-13 deal with $5$-variables.
Steps~14-15 deal with $4$-variables.
If the maximum degree of the formula is $3$, then we apply the algorithm by Wahlstr{\"{o}}m~\cite{DBLP:conf/sat/Wahlstrom05}, which is Step 16 of our algorithm.

\begin{algorithm*}[t!]\label{main-algorithm}
    \caption{SAT($\mathcal{F}$)}
    \KwIn{a CNF-formula $\mathcal{F}$}
    \KwOut{$1$ or $0$ to indicate the satisfiability of $\mathcal{F}$}

    \textbf{Step 1.} If $\mathcal{F}=\emptyset$, return 1. Else if $\mathcal{F}$ contains an empty clause, return 0.\\
    \textbf{Step 2.} If $\mathcal{F}$ is not a reduced CNF-formula, iteratively apply the reduction rules to reduce it.\\
    \textbf{Step 3.} If the degree of $\mathcal{F}$ is at least $6$, select a variable $x$ with the maximum degree and return SAT($\mathcal{F}_{x=1}$)$\vee$SAT($\mathcal{F}_{x=0}$).\\
    \textbf{Step 4.} If there is a $(1,4)$-literal $x$ (assume that $xC$ is the unique clause containing $x$), return SAT($\mathcal{F}_{x=1\And C=0}$)$\vee$SAT($\mathcal{F}_{x=0}$).\\
    \textbf{Step 5.} If there is a $5$-literal $x$ such that at least two $2$-clauses contain $x$ or $\overline{x}$, return SAT($\mathcal{F}_{x=1}$)$\vee$SAT($\mathcal{F}_{x=0}$).\\
    \textbf{Step 6.} If there are two $5$-literals $x$ and $y$ contained in one $2$-clause $xy$, return SAT($\mathcal{F}_{x=1}$)$\vee$SAT($\mathcal{F}_{x=0}$).\\
    \textbf{Step 7.} If there is a $5$-literal $x$ contained in a $2$-clause, return SAT($\mathcal{F}_{x=1}$)$\vee$SAT($\mathcal{F}_{x=0}$).\\
    \textbf{Step 8.} If there is a $5$-literal $x$ such that $N(x, \mathcal{F})$ and $N(\overline{x}, \mathcal{F})$ contain at least two $4^-$-literals, return SAT($\mathcal{F}_{x=1}$)$\vee$SAT($\mathcal{F}_{x=0}$).\\
    \textbf{Note:} If there are still some 5-literals, they must be (2,3)/(3,2)-literals. In the next two steps, we let $x$ be a $(2, 3)$-literal and $xC_1, xC_2, \overline{x}D_1, \overline{x}D_2$, and $\overline{x}D_3$ be the five clauses containing $x$ or $\overline{x}$.\\
    \textbf{Step 9.} If there exist $5$-literals $y_1$ and $y_2$ such that $y_1\in C_1$, $y_1\in D_1$, $y_2\in C_2$ and $y_2$ or $\overline{y_2}\in D_2$, return SAT($\mathcal{F}_{y_1=1}$)$\vee$SAT($\mathcal{F}_{y_1=0}$).\\
    \textbf{Step 10.} If there exist $5$-literals $y_1$  and $y_2$ such that $y_1\in C_1$, $\overline{y_1}\in D_1$, $y_2\in C_2$, and $\overline{y_2}\in D_2$, pick a $5$-literal $z\in D_3$ (let $R_5(\mathcal{F}_{z=1})$ denote the resulting formula after only applying R-Rule~\ref{rule-trivial-resol} on $\mathcal{F}_{z=1}$) and return SAT($R_5(\mathcal{F}_{z=1})$)$\vee$SAT($\mathcal{F}_{z=0}$).\\
    \textbf{Step 11.} If there is a $5$-literal $x$ contained in at least one $4^+$-clause, return SAT($\mathcal{F}_{x=1}$)$\vee$SAT($\mathcal{F}_{x=0}$).\\
    \textbf{Step 12.} If there is a clause containing both a $5$-literal $x$ and a $4^-$-literal, return SAT($\mathcal{F}_{x=1}$)$\vee$SAT($\mathcal{F}_{x=0}$).\\
    \textbf{Step 13.} If there are still some 5-literals, then $\mathcal{F}=\mathcal{F}_{5} \wedge \mathcal{F}_{\leq4}$, where $\mathcal{F}_{5}$ is a 3-CNF containing only 5-literals and $\mathcal{F}_{\leq 4}$ contains only 3/4-literals.
    We solve $\mathcal{F}_{5}$ by using the $3$-SAT algorithm by Liu~\cite{DBLP:conf/icalp/Liu18} (let $A(\mathcal{F}_{5})$ denote the result) and return $A(\mathcal{F}_{5})$ $\wedge$ SAT($\mathcal{F}_{\leq 4}$).\\
    \textbf{Step 14.} If there is a $(1,3)$-literal $x$ (assume that $xC$ is the unique clause containing $x$), return SAT($\mathcal{F}_{x=1\And C=0}$)$\vee$SAT($\mathcal{F}_{x=0}$).\\
    \textbf{Step 15.} If there is a $(2,2)$-literal $x$, return SAT($\mathcal{F}_{x=1}$)$\vee$SAT($\mathcal{F}_{x=0}$).\\
    \textbf{Step 16.} Apply the algorithm by Wahlstr{\"{o}}m~\cite{DBLP:conf/sat/Wahlstrom05} to solve the instance.
\end{algorithm*}

Before analyzing the algorithm, we compare our algorithm with the previous algorithm by Chen and Liu~\cite{DBLP:conf/wads/ChenL09}.
We can see that they used a simple and uniform branching rule to deal with variables of degree at least $5$ and used careful and complicated branching rules for $4$-variables. Their bottlenecks contain one case of branching on $5$-variables and one case of dealing with $4$-variables. We carefully design and analyze the branching rules for $5$-variables to avoid one previous bottleneck and also simplify the branching rules for $4$-variables.
We use Steps~4-13 to deal with different structures of $5$-variables, while reference~\cite{DBLP:conf/wads/ChenL09} just used a single step like our Step~3 to deal with $5$-variables.
Steps~14-15 are simple branching rules to deal with $4$-variables, while reference~\cite{DBLP:conf/wads/ChenL09} used complicated rules to deal with 4-variables.
To get further improvements, we may also check some special structures and propose branching rules to deal with them so that we can prove the worst case in our algorithm would not always happen.

We analyze the correctness of the algorithm.
Step~2 applies reduction rules, and their correctness has been proven previously.
Step~3 deals with variables with degree $\geq 6$.
Steps~4-13 deal with $5$-variables by simple branching.
After Step~7, all clauses containing $5$-literals are $3^+$-clause.
After Step~8, there is at most one $4^-$-literal in the neighbor set of $x$ or $\overline{x}$.
In Step~10, there must be a $5$-literal in $D_3$ since $|D_3|\geq 2$, and there is at most one $4^-$-literal in it. Thus, the condition of Step~10 holds.
If Steps~$1$-$12$ do not apply, then $\mathcal{F}$ can be written as $\mathcal{F}=\mathcal{F}_5 \wedge \mathcal{F}_{\leq4}$, where $\mathcal{F}_{5}$ is a 3-CNF with $var(\mathcal{F}_{5})$ be the set of 5-variables in $\mathcal{F}$ and $var(\mathcal{F}_{5})\cap var(\mathcal{F}_{\leq4}) = \emptyset$.
So we can do Step~13.
Steps~14-15 deal with $4$-variables by branching.
When the algorithm comes to the last step, all variables must have a degree of $3$, and we apply the algorithm by Wahlstr{\"{o}}m~\cite{DBLP:conf/sat/Wahlstrom05} to deal with this special case.
The hard part is to analyze the running time bound, which will be presented below.

\section{The Analysis Framework}

We use the measure-and-conquer method to analyze the running time bound of our algorithm and adopt $\mu(\mathcal{F})$ defined in (\ref{measure}) as the measure to construct recurrence relations for our branching operations.
Before analyzing each detailed step of the algorithm, we first introduce some general frameworks of our analysis and prove some lemmas that will be used as lower bounds in the detailed step analyses.

In each sub-branch of a branching operation, we assign value $1$ or $0$ to some literals and remove some clauses and literals. If we assign value $1$ to a literal $x$ in the formula $\mathcal{F}$, then we will remove all clauses containing $x$ from the formula since all those clauses are satisfied. We also remove all $\overline{x}$ literals from the clauses containing $\overline{x}$ since those literals get value $0$. The assignment and removing together are called an \emph{assignment operation}. We may assign values to more than one literal, and we do assignment operations for each literal.

Let $S$ be a subset of literals. We use  $\mathcal{F}_{S=1}$ to denote the resulting formula after assigning 1 to each literal in $S$ and doing assignment operations for each literal in $S$. Note that $\mathcal{F}_{S=1}$ may not be a reduced formula, and we will apply our reduction rules to reduce it. We use $\mathcal{F}'_{S=1}$ to denote the reduced formula obtained from $\mathcal{F}_{S=1}$, i.e., $\mathcal{F}'_{S=1}=R(\mathcal{F}_{S=1})$, and use $\mathcal{F}^{*}_{S=1}$ to denote the first formula during we apply reduction rules on $\mathcal{F}_{S=1}$ such that R-Rule~$1$-$4$ are not applicable on the formula. We analyze how much we can reduce the measure in each sub-branch by establishing some lower bounds for
\begin{equation*}
    \Delta_{S}^* = \mu(\mathcal{F}) - \mu(\mathcal{F}_{S=1}^*);
\end{equation*}
\begin{equation*}
    \Delta_{S} = \mu(\mathcal{F}) - \mu(\mathcal{F}_{S=1}').
\end{equation*}

For the sake of presentation, we define
\begin{equation*}
    \xi_{S}^{(1)} = \mu(\mathcal{F}) - \mu(\mathcal{F}_{S=1});
\end{equation*}
\begin{equation*}
    \xi_{S}^{(2)} = \mu(\mathcal{F}_{S=1}) - \mu(\mathcal{F}_{S=1}^*);
\end{equation*}
\begin{equation*}
    \xi_{S}^{(3)} = \mu(\mathcal{F}_{S=1}^*) - \mu(\mathcal{F}_{S=1}').
\end{equation*}
Thus, it holds that
\[
    \Delta_{S}^* = \mu(\mathcal{F}) - \mu(\mathcal{F}_{S=1}^*) = \mu(\mathcal{F}) - \mu(\mathcal{F}_{S=1}) + \mu(\mathcal{F}_{S=1}) - \mu(\mathcal{F}_{S=1}^*) = \xi_{S}^{(1)} + \xi_{S}^{(2)}
\]
and
\[
    \Delta_{S} = \mu(\mathcal{F}) - \mu(\mathcal{F}_{S=1}') = \mu(\mathcal{F}) - \mu(\mathcal{F}_{S=1}^*) + \mu(\mathcal{F}_{S=1}^*) - \mu(\mathcal{F}_{S=1}') = \Delta_{S}^* + \xi_{S}^{(3)}.
\]

In other words, $\xi_{S}^{(1)}$ is the amount of measure of $\mathcal{F}$ reduced by only doing assignment operations for each literal in $S$;
$\xi_{S}^{(2)}$ is the amount of measure of $\mathcal{F}_{S=1}$ reduced by the first four reduction rules;
$\xi_{S}^{(3)}$ is the amount of measure of $\mathcal{F}_{S=1}^*$ reduced by iteratively applying the reduction rules until it becomes a reduced formula.

We analyze the running time of our algorithm by analyzing the branching vector/factor generated by each step.
In each branching step, we will branch into two sub-branches.
Assume that all literals in $S_1$ are assigned the value $1$ in the first sub-branch, and all literals in $S_2$ are assigned the value $1$ in the second sub-branch.
In the analysis, we will frequently use the following property.

\begin{lemma}\label{lemma-bvec}
It holds that
\begin{enumerate}[(i)]
    \item if we can show that $\Delta_{S_1}\geq p$ and $\Delta_{S_2}\geq q$, then the branching vector generated by this branching rule is covered by $[p, q]$;
    \item if we can show $\min (\Delta_{S_1}, \Delta_{S_2}) \geq a$ and $\Delta_{S_1}+ \Delta_{S_2} \geq b$, then the branching vector generated by this branching rule is covered by $[a, b-a]$.
\end{enumerate}
\end{lemma}

Next, we will analyze some lower bounds for $\Delta_{S_1}, \Delta_{S_2}$ and $\Delta_{S_1}+\Delta_{S_2}$.

\medskip

According to the assignment operation, we know that all variables of the literals in $S$ will not appear in $\mathcal{F}_{S=1}$. So we have a
trivial bound
\begin{equation*}
    \xi_{S}^{(1)} \geq \sum_{v\in S}w_{deg(v)}.
\end{equation*}

To get better bounds, we first define some notations.
Recall that in a formula $\mathcal{F}$, a literal $x$ is called a neighbor of a literal $z$ if there is a clause containing both $z$ and $x$.
We use $N(z, \mathcal{F})$ to denote the set of neighbors of a literal $z$ in a formula $\mathcal{F}$ and $N^{(k)}(x, \mathcal{F})$ (resp., $N^{(k+)}(z, \mathcal{F})$) to denote the neighbors of $z$ in $k$-clauses (resp., $k^+$-clauses) in $\mathcal{F}$.
\begin{definition}

For a literal $x$ in a reduced formula $\mathcal{F}$ and $i\in N^+$, we define the following notations:
\begin{itemize}
    \item $n_i(x)$: the number of literals with degree $i$ that appear in $N(x, \mathcal{F})$,
    i.e., $n_i(x) = |\{y: y\in N(x, \mathcal{F})\text{ and } deg(y)=i\}|$;
    \item $n'_i(x)$: the number of literals with degree $i$ that appear in $N^{(2)}(x, \mathcal{F})$,
    i.e., $n'_i(x) = |\{y: y\in N^{(2)}(x, \mathcal{F})\text{ and } deg(y)=i\}|$;
    \item $n''_i(x)$: the number of literals with degree $i$ that appear in $N^{(3+)}(x, \mathcal{F})$,
    i.e., $n''_i(x) = |\{y: y\in N^{(3+)}(x, \mathcal{F})\text{ and } deg(y)=i\}|$;
    \item $t_{i,1}(x)$: the number of $i$-variables $y$ such that only one of $y$ and $\overline {y}$ appears in $N(x, \mathcal{F})$,
    i.e., $t_{i,1}(x)=|\{var(y):|\{y,\overline{y}\}\cap N(x, \mathcal{F})|=1\}|$;
    \item $t_{i,2}(x)$: the number of $i$-variables $y$ such that both of $y$ and $\overline {y}$ appear in $N(x, \mathcal{F})$,
    i.e., $t_{i,2}(x)=|\{var(y):|\{y,\overline{y}\}\cap N(x, \mathcal{F})|=2\}|$.
\end{itemize}
\end{definition}

\noindent \emph{Example.} Let $\mathcal{F} = (xy_1z_1)\wedge(xy_2\overline{z_1}l_1)\wedge(xy_3\overline{z_4})\wedge(\overline{x}z_2z_3\overline{l_2})\wedge(\overline{x}z_4)\wedge \mathcal{F}_1$, where $x, z_1, z_2, z_3, z_4$ are $5$-variables, $y_1, y_2, y_3$ are $4$-variables, and $l_1, l_2$ are $3$-variables.
Then, in this case, we have:
\begin{itemize}
    \item $n_3(x) = n''_3(x) = t_{3,1}(x) = |\{l_1\}| = 1$;
    \item $n_4(x) = n''_4(x) = t_{4,1}(x) = |\{y_1,y_2,y_3\}| = 3$;
    \item $n_5(x) = n''_5(x) = |\{z_1,\overline{z_1},z_3\}| = 3$;
    \item $t_{5,1}(x) = |\{z_3\}| = 1$ and $t_{5,2}(x) = |\{z_1\}| = 1$;
    \item $n_4(\overline{x}) = |\{z_2,z_3,z_4\}|=3$, $n'_4(\overline{x}) = |\{z_4\}| = 1$, and $n''_4(\overline{x}) = |\{z_2,z_3\}| = 2$.
\end{itemize}
In a reduced formula $\mathcal{F}$, there is no $1$-clause since R-Rule~\ref{rule-pure} is not applicable. So for a literal $x$, it holds that
\[
    n_i(x)=n'_i(x)+n''_i(x).
\]
By Lemma~\ref{prop4}, for a literal $x$, we know that all literals in $N(x, \mathcal{F})$ are different from each other. So for any variable $y$, there is at most one literal $y$ and at most one literal $\overline{y}$ appearing in $N(x, \mathcal{F})$, which implies that
\[
    n_i(x) = t_{i,1}(x) + 2t_{i,2}(x).
\]

Next, we give some lower bounds on $\xi_{S}^{(1)}$, $\xi_{S}^{(2)}$, and  $\Delta_{S_1}^*+\Delta_{S_2}^*$, which will be used to prove our main results.

The following lemma shows how much the measure decreases after we do the assignment operation to a single literal.
\begin{lemma}\label{dec-assign2}
    Assume that $\mathcal{F}$ is a reduced CNF-formula with degree $d$. Let $S = \{x\}$, where $x$ is a literal with degree $d$ in $\mathcal{F}$. It holds that
    \begin{equation*}
        \xi_{S}^{(1)} \geq w_d + \sum_{3\leq i\leq d}n_i(x)\delta_i.
    \end{equation*}
\end{lemma}
\begin{proof}

    After assigning value 1 to literal $x$, all clauses containing $x$ will be removed from $\mathcal{F}$, that is, all literals in $N(x, \mathcal{F})$ will be removed from $\mathcal{F}$.
    By Lemma \ref{prop4}, all literals in $N(x, \mathcal{F})$ are different from each other.
    For a variable $y$, there are at most two literals of it ($y$ and $\overline{y}$) in $N(x, \mathcal{F})$. So the degree of variable $y$ will decrease at most 2.
    Moreover, if $y$ is a $3$-variable, by Lemma \ref{prop4}, we know that $y$ and $\overline{y}$ will not simultaneously appear in $N(x, \mathcal{F})$ since there are no other clauses containing $x\overline{y}$ (i.e., $t_{3,2}(x) = 0$). So the degree of $y$ will decrease by at most $1$ if $y$ is a $3$-variable.

    Note that $\xi_{S}^{(1)}$ expresses how much the measure of $\mathcal{F}$ decreases after the assignment operation. So it holds that
    \[
        \xi_{S}^{(1)} \geq w_d + \sum_{3\leq i\leq d}{t_{i,1}(x)\delta_i} + \sum_{4\leq i\leq d}{t_{i,2}(x)(\delta_i+\delta_{i-1})}.
    \]
    Since $t_{3,2}(x)=0$, can write it as
    \[
        \xi_{S}^{(1)} = w_d + \sum_{3\leq i\leq d}{t_{i,1}(x)\delta_i} + \sum_{3\leq i\leq d}{t_{i,2}(x)(\delta_i+\delta_{i-1})}.
    \]
    By $\delta_i\geq \delta_{i-1}$ and $\delta_i\geq \delta_d$ for $3\leq i\leq d$, we have
    \begin{align*}
        \xi_{S}^{(1)}
         & \geq w_d + \sum_{3\leq i\leq d}{t_{i,1}(x)\delta_i} + \sum_{3\leq i\leq d}{t_{i,2}(x)(\delta_i + \delta_i)} \\
         & = w_d + \sum_{3\leq i\leq d}{(t_{i,1}(x) + 2t_{i,2}(x))\delta_i}.
    \end{align*}
    Note that $n_i(x)=t_{i,1}(x) + 2t_{i,2}(x)$ holds for $3\leq i\leq d$. We finally obtain
    \[
        \xi_{S}^{(1)} = w_d + \sum_{3\leq i\leq d}{n_i(x)\delta_i}.
    \]

\end{proof}

\begin{lemma}\label{coro-assign2} 
    Assume that $\mathcal{F}$ is a reduced CNF-formula with degree $d$. Let $S = \{x\}$, where $x$ is a $(j, d-j)$-literal in $\mathcal{F}$. It holds that
    \begin{equation*}
        \xi_{S}^{(1)} \geq w_d + j\delta_d.
    \end{equation*}
\end{lemma}
\begin{proof}
    By Lemma \ref{dec-assign2} and $\delta_i\geq \delta_d$ for $3\leq i\leq d$, we get
    \begin{align*}
        \xi_{S}^{(1)}
        \geq w_{deg(x)} + \sum_{3\leq i\leq d}n_i(x)\delta_i
        \geq w_d + d\sum_{3\leq i\leq d}n_i(x)\delta_d.
    \end{align*}
    As $x$ is a $(j,d-j)$-literal, there are $j$ clauses containing literal $x$.
    Since there is no $1$-clause in $\mathcal{F}$ (all $1$-clauses are reduced by R-Rule~\ref{rule-pure}), literal $x$ has at least $j$ neighbors, i.e., $\sum_{3\leq i\leq d}n_i(x) \geq j$.
    We further obtain
    \begin{align*}
        \xi_{S}^{(1)} \geq w_d + j\delta_d.
    \end{align*}
\end{proof}

\begin{lemma}\label{dec-reduce}
    Assume that $\mathcal{F}$ is a reduced CNF-formula with degree $d$. Let $S = \{x\}$, where $x$ is a literal in $\mathcal{F}$. It holds that
    \begin{equation*}
        \xi_{S}^{(2)}\geq n'_3(\overline{x})w_3 + \sum_{4\leq i\leq d}n'_i(\overline{x})w_{i-1}.
    \end{equation*}
\end{lemma}
\begin{proof}

    Recall that $\xi_{S}^{(2)}$ expresses how much the measure decreases after applying reduction Rule 1-4 on $\mathcal{F}_{S=1}$. All literals in $N^{(2)}(\overline{x}, \mathcal{F})$ will be assigned value 1 since they will be contained in $1$-clauses in $\mathcal{F}_{S=1}$ and R-Rule \ref{rule-pure} is applied.

    For a literal $y\in N^{(2)}(\overline{x}, \mathcal{F})$, we know that there are no other clauses containing $xy$ (i.e., $y\notin N(x, \mathcal{F})$) by Lemma \ref{prop2}. So if $var(y)$ is an $i$-variable in $\mathcal{F}$, then the degree of $y$ in $\mathcal{F}_{S=1}$ is at least $i-1$.
    Moreover, for a literal $y\in N^{(2)}(x, \mathcal{F})$ such that $var(y)$ is a $3$-variable, by Lemma \ref{prop4} we know that there are no other clauses containing $xy$ or $x\overline{y}$ (i.e., $y, \overline{y}\notin N(x, \mathcal{F})$). So the degree of $y$ would still be $3$ in $\mathcal{F}_{S=1}$.

    Since all variables corresponding to the literals in $N(\overline{x}, \mathcal{F})$ would not appear in $\mathcal{F}_{S=1}'$ and these variables are different from each other by Lemma \ref{prop2}, we obtain
    \begin{equation*}
        \xi_{S}^{(2)}\geq n'_3(\overline{x})w_3 + \sum_{4\leq i\leq d}n'_i(\overline{x})w_{i-1}.
    \end{equation*}
\end{proof}

After getting the lower bounds on $\xi_{S}^{(1)}$ and $\xi_{S}^{(2)}$, we are going to establish some lower bounds when we branch on a single literal $x$: $\Delta_{S}^*$ for the case $S=\{x\}$ and $\Delta_{S_1}^* + \Delta_{S_2}^*$ for the case that $S_1=\{x\}$ and $S_2=\{\overline{x}\}$.
We first consider a general lower bound on $\Delta_{S_1}^*+\Delta_{S_2}^*$.

\begin{lemma}\label{dec-sum}
    Assume that $\mathcal{F}$ is a reduced CNF-formula of degree $d$. Let $S_1 = \{x\}$ and $S_2 = \{\overline{x}\}$, where $x$ is a $d$-variable in $\mathcal{F}$. It holds that
    \begin{align*}
        \Delta_{S_1}^*+\Delta_{S_2}^*
         & \geq 2w_d + 2d\delta_d + (n'_3(x)+n'_3(\overline{x}))(2w_3-2\delta_d)    \\
         & \quad+ \sum_{4\leq i\leq d}{(n'_i(x)+n'_i(\overline{x}))(w_i-2\delta_d)}.
    \end{align*}
\end{lemma}
\begin{proof}
    By Lemma \ref{dec-assign2} and Lemma \ref{dec-reduce}, we have
    \begin{align*}
        \xi_{S_1}^{(1)} \geq w_d + \sum_{3\leq i\leq d}n_i(x)\delta_i
        ~~\text{and}~~
        \xi_{S_1}^{(2)}\geq n'_3(\overline{x})w_3 + \sum_{4\leq i\leq d}n'_i(\overline{x})w_{i-1}.
    \end{align*}
    We can get lower bounds of $\xi_{S_2}^{(1)}$ and $\xi_{S_2}^{(2)}$ by the same way, and then we have:
    \begin{equation}\label{aux-bound-1}
        \begin{aligned}
            \Delta_{S_1}^* & = \xi_{S_1}^{(1)} + \xi_{S_1}^{(2)} \geq w_d + \sum_{3\leq i\leq d}n_i(x)\delta_i + n'_3(\overline{x})w_3 + \sum_{4\leq i\leq d}n'_i(\overline{x})w_{i-1}; \\
            \Delta_{S_2}^* & = \xi_{S_2}^{(1)} + \xi_{S_2}^{(2)} \geq w_d + \sum_{3\leq i\leq d}n_i(\overline{x})\delta_i + n'_3(x)w_3 + \sum_{4\leq i\leq d}n'_i(x)w_{i-1}.
        \end{aligned}
    \end{equation}
    As $\sum_{3\leq i\leq d}(n'_i(x)+n'_i(\overline{x}))$ is the number of $2$-clauses containing $x$ or $\overline{x}$ and there are $d$ clauses containing $x$ or $\overline{x}$, we know that the number of $3^+$-clauses that contains $x$ or $\overline{x}$ is $d-\sum_{3\leq i\leq d}(n'_i(x)+n'_i(\overline{x}))$. Recall that $n''_i(x)$ is the number of literals with degree $i$ that appear in $N^{(3+)}(x, \mathcal{F})$. We get
    \begin{equation}\label{ineq-n}
        \sum_{3\leq i\leq d}(n''_i(x)+n''_i(\overline{x}))\geq 2(d-\sum_{3\leq i\leq d}(n'_i(x)+n'_i(\overline{x}))).
    \end{equation}
    By (\ref{aux-bound-1}) and summing $\Delta_{S_1}^*$ and $\Delta_{S_2}^*$ up, we have
    \begin{align*}
        \Delta_{S_1}^*+\Delta_{S_2}^*
        & \geq w_d + \sum_{3\leq i\leq d}n_i(x)\delta_i + n'_3(\overline{x})w_3 + \sum_{4\leq i\leq d}n'_i(\overline{x})w_{i-1} \\
        & \quad + w_d + \sum_{3\leq i\leq d}n_i(\overline{x})\delta_i + n'_3(x)w_3 + \sum_{4\leq i\leq d}n'_i(x)w_{i-1}         \\
        & = 2w_d + (n'_3(x)+n'_3(\overline{x}))w_3 + \sum_{4\leq i\leq d}(n'_i(x)+n'_i(\overline{x}))w_{i-1}                      \\
        & \quad + \sum_{3\leq i\leq d}(n_i(x)+n_i(\overline{x}))\delta_i.
    \end{align*}
    Next, we first consider the lower bound on the term $\sum_{3\leq i\leq d}(n_i(x)+n_i(\overline{x}))\delta_i$ and then further analyze $\Delta_{S_1}^*+\Delta_{S_2}^*$.

    By $\delta_i\geq \delta_d$ for $3\leq i\leq d$ and (\ref{ineq-n}), we have
    \begin{align*}
        \sum_{3\leq i\leq d}(n''_i(x)+n''_i(\overline{x}))\delta_i
        & \geq (\sum_{3\leq i\leq d}(n''_i(x)+n''_i(\overline{x})))\delta_d
        \\
        & \geq 2(d-\sum_{3\leq i\leq d}(n'_i(x)+n'_i(\overline{x})))\delta_d.
    \end{align*}
    With $n_i(x)=n'_i(x)+n''_i(x)$ and $n_i(\overline{x})=n'_i(\overline{x})+n''_i(\overline{x})$, we get
    \begin{align*}
        \sum_{3\leq i\leq d}(n_i(x)+n_i(\overline{x}))\delta_i
        & = \sum_{3\leq i\leq d}(n'_i(x)+n'_i(\overline{x}))\delta_i + \sum_{3\leq i\leq d}(n''_i(x)+n''_i(\overline{x}))\delta_i\\
        &\geq \sum_{3\leq i\leq d}(n'_i(x)+n'_i(\overline{x}))\delta_i + 2(d-\sum_{3\leq i\leq d}(n'_i(x)+n'_i(\overline{x})))\delta_d\\
        &\geq 2d\delta_d + \sum_{3\leq i\leq d}(n'_i(x)+n'_i(\overline{x}))(\delta_i-2\delta_d).
    \end{align*}
    Next we continue the previous analysis on $\Delta_{S_1}^*+\Delta_{S_2}^*$. By applying the above result, we have
    \begin{align*}
        \Delta_{S_1}^*+\Delta_{S_2}^*
        & \geq 2w_d + (n'_3(x)+n'_3(\overline{x}))w_3 + \sum_{4\leq i\leq d}(n'_i(x)+n'_i(\overline{x}))w_{i-1}                            \\
        & \quad + \sum_{3\leq i\leq d}(n_i(x)+n_i(\overline{x}))\delta_i                                                                   \\
        & = 2w_d + (n'_3(x)+n'_3(\overline{x}))w_3 + \sum_{4\leq i\leq d}(n'_i(x)+n'_i(\overline{x}))w_{i-1}                  \\
        & \quad + 2d\delta_d + \sum_{3\leq i\leq d}(n'_i(x)+n'_i(\overline{x}))(\delta_i-2\delta_d).
    \end{align*}
    Expanding $3\leq i\leq d$ to $i=3$ and $4\leq i\leq d$ in the last term and then combining like terms, we get
    \begin{align*}
        \Delta_{S_1}^*+\Delta_{S_2}^*
         & = 2w_d + 2d\delta_d + (n'_3(x)+n'_3(\overline{x}))w_3 + \sum_{4\leq i\leq d}(n'_i(x)+n'_i(\overline{x}))w_{i-1}                  \\
         & \quad + (n'_3(x)+n'_3(\overline{x}))(\delta_3-2\delta_d) + \sum_{4\leq i\leq d}(n'_i(x)+n'_i(\overline{x}))(\delta_i-2\delta_d) \\
         & \geq 2w_d + 2d\delta_d + (n'_3(x)+n'_3(\overline{x}))(w_3+\delta_3-2\delta_d)        \\
         & \quad + \sum_{4\leq i\leq d}(n'_i(x)+n'_i(\overline{x}))(w_{i-1}+\delta_i-2\delta_d) \\
         & = 2w_d + 2d\delta_d + (n'_3(x)+n'_3(\overline{x}))(2w_3-2\delta_d)                   \\
         & \quad + \sum_{4\leq i\leq d}{(n'_i(x)+n'_i(\overline{x}))(w_i-2\delta_d)}.
    \end{align*}
\end{proof}

In our algorithm, we apply a stronger branching for a $(1, 4)/(1, 3)$-literal $x$.
Assume $xC$ is the unique clause containing literal $x$.
The following two lemmas show lower bounds on $\Delta_S$ for the case $S=\{x\}\cup \overline{C}$ and $\Delta_{S_1}+\Delta_{S_2}$ for the case $S_1=\{x\}\cup \overline{C}$ and $S_2 = \{\overline{x}\}$.

\begin{lemma}\label{dec-1ltr}
    Assume that $\mathcal{F}$ is a reduced CNF-formula of degree $d=4$ or $5$. Let $x$ be a $(1, d-1)$-literal and $xC$ be the
    unique clause containing $x$ in $\mathcal{F}$. Let $S_1=\{x\}\cup \overline{C}$ and $S_2=\{\overline{x}\}$.
    It holds that
    \[
        \Delta_{S_1} + \Delta_{S_2} \geq 2w_d+3w_3+(2d-3)\delta_d
    \]
    and
    \[
        \min(\Delta_{S_1}, \Delta_{S_2})\geq w_d+\min(2w_3, (d-1)\delta_d).
    \]
\end{lemma}
\begin{proof} 
    We first consider $\Delta_{S_1}$.
    By Lemma \ref{prop5}, we know that $|C|\geq 2$ and $var(C)\cap var(N^{(2)}(\overline{x}, \mathcal{F})) = \emptyset$. So, after assigning value 1 to all literals in $S_1$, all $2$-clauses containing $\overline{x}$ in $\mathcal{F}$ would become $1$-clauses in $\mathcal{F}_{S_1=1}$ and the degree of each variable in these $1$-clauses is the same as that in $\mathcal{F}$. Then, by applying R-Rule~\ref{rule-pure}, all variables in these $1$-clauses get assignments. Thus, all variables in $var(C)\cup var(N^{(2)}(\overline{x}, \mathcal{F}))$ would not appear in $\mathcal{F}_{S_1=1}'$. Since all variables in $C$ are $3^+$-variables and $w_3\leq w_i$ for $i\geq 3$, we preliminarily have
    \begin{equation}\label{dec-1ltr-loose}
        \begin{aligned}
            \Delta_{S_1}
            & \geq \sum_{v\in S_1}w_{deg(v)} + \sum_{3\leq i\leq d}n'_i(\overline{x})w_i\\
            & \geq w_d+|C|w_3 + \sum_{3\leq i\leq d}n'_i(\overline{x})w_i\\
            & \geq w_d+2w_3 + \sum_{3\leq i\leq d}{n'_i(\overline{x})w_i}.
        \end{aligned}
    \end{equation}
    Note that when we assign value 1 to a literal $y\in \overline{C}$, the neighbors of $y$ would also be removed from the formula, which may further decrease the measure if $y$ is not in $N^{(2)}(\overline{x}, \mathcal{F})$ (Otherwise, for a literal $l\in N(y, \mathcal{F}) \cup N^{(2)}(\overline{x}, \mathcal{F})$, it will get assignment by applying R-Rule~\ref{rule-pure} and the decrease of measure will be counted in $\sum_{3\leq i\leq d}n'_i(\overline{x})w_i$).
    So we define $N_c=\{z: z\in \bigcup_{y\in \overline{C}} N(y, \mathcal{F}) \text{~and~} z\notin N^{(2)}(\overline{x}, \mathcal{F})\}$.
    As $\delta_d \leq \delta_i$ for $i\leq d$, we further get
    \begin{align*}
        \Delta_{S_1}
        & \geq w_d + 2w_3 + \sum_{3\leq i\leq d}n'_i(\overline{x})w_i + |N_c|\delta_d\\
        & \geq w_d + 2w_3 + (\sum_{3\leq i\leq d}n'_i(\overline{x}))w_3 + |N_c|\delta_d.
    \end{align*}
    For $\Delta_{S_2}$, by Lemma~\ref{dec-assign2} and $\delta_d\leq \delta_i$ for $3\leq i\leq d$, we have
    \[
        \Delta_{S_2}
        \geq w_d+\sum_{3\leq i\leq d}n_i(\overline{x})\delta_i
        \geq w_d+(\sum_{3\leq i\leq d}n_i(\overline{x}))\delta_d.
    \]
    Let $p = \sum_{3\leq i\leq d}n'_i(\overline{x})$ and $q = \sum_{3\leq i\leq d}n''_i(\overline{x})$.
    Note that $p$ is also the number of $2$-clauses containing $\overline{x}$ and $0\leq p\leq d-1$.
    Since there are $d-1$ clauses containing $\overline{x}$, it holds that $q\geq 2(d-1-p)$.
    Recall that $n_i(\overline{x})=n'_i(\overline{x})+n''_i(\overline{x})$ for $i\in N^+$, we have
    \[
        \sum_{3\leq i\leq d}n_i(\overline{x}) = p + q\geq p + 2(d-1-p) = 2d-2-p.
    \]
    So
    \[
        \Delta_{S_1}\geq w_d + 2w_3 + pw_3 + |N_c|\delta_d \text{~~and~~}
        \Delta_{S_2}\geq w_d+(2d-2-p)\delta_d.
    \]
    Summing them up, we get
    \[
        \Delta_{S_1}+\Delta_{S_2} \geq 2w_d+2w_3 + (|N_c| + 2d-2)\delta_d + p(w_3-\delta_d).
    \]
    Note that since $0\leq p\leq d-1$, it holds that $\Delta_{S_2}\geq w_d+(d-1)\delta_d$.

    Next, let us consider the following two cases.

    \textbf{Case 1.} $\overline{x}$ is contained in at least one $2$-cluase, i.e., $p\geq 1$. We have
    \[
        \Delta_{S_1} + \Delta_{S_2} \geq 2w_d + 2w_3 + (2d-2)\delta_d + w_3-\delta_d = 2w_d+3w_3+(2d-3)\delta_d.
    \]
    For this case, we also have $\Delta_{S_1}\geq w_d + 3w_3$, $\Delta_{S_2}\geq w_d+(d-1)\delta_d$, and $\min(\Delta_{S_1}, \Delta_{S_2})\geq w_d + \min(3w_3, (d-1)\delta_d)\geq 2w_3$ since that $2w_3<(d-1)\delta_d$ holds for $d=4$ and $d=5$.

    \textbf{Case 2.} All clauses containing $\overline{x}$ are $3^+$-clauses, i.e., $N^{(2)}(\overline{x}, \mathcal{F}) = \emptyset$ and $p=0$.
    This implies $|N_c|\geq 1$ since $\bigcup_{y\in \overline{C}} N(y, \mathcal{F})\neq \emptyset$. Note that $w_3<2\delta_d$ holds for $d=4$ and $d=5$.
    We have
    \begin{align*}
        \Delta_{S_1} + \Delta_{S_2} \geq 2w_d + 2w_3 + (1 + 2d-2)\delta_d
        &= 2w_d + 2w_3 + (2d-1)\delta_d \\
        &\geq 2w_d + 3w_3 + (2d-3)\delta_d.
    \end{align*}

    For this case, we also have $\Delta_{S_1}\geq w_d + 2w_3$, $\Delta_{S_2}\geq w_d+(d-1)\delta_d$, and $\min(\Delta_{S_1}, \Delta_{S_2})\geq w_d + \min(2w_3, (d-1)\delta_d)\geq 2w_3$ since that $2w_3<(d-1)\delta_d$ holds for $d=4$ and $d=5$.
\end{proof}

As shown in Algorithm~\ref{main-algorithm}, we consider several cases for $5$-litreals. The following lemma is a corollary based on Lemma \ref{dec-assign2} in order to get tighter bounds on $\Delta_{S_1}+\Delta_{S_2}$ for the case $S_1=\{x\}$ and $S_2=\{\overline{x}\}$ where $x$ is a $5$-variable.
\begin{lemma}\label{dec-sum-5d3c-t}
    Assume that $\mathcal{F}$ is a reduced CNF-formula of $d=5$. Let $S_1=\{x\}$ and $S_2=\{\overline{x}\}$, where $x$ is a $(2,3)/(3,2)$-literal in $\mathcal{F}$. If all clauses containing $x$ or $\overline{x}$ are $3^+$-clauses, it holds that
    \begin{align*}
        \xi_{S_1}^{(1)}+\xi_{S_2}^{(1)} \geq 2w_5 + \sum_{3\leq i\leq 5}(n_i(x)+n_i(\overline{x}))\delta_i + (t_{5,2}(x)+t_{5,2}(\overline{x}))(\delta_4-\delta_5).
    \end{align*}
\end{lemma}
\begin{proof} 
    From the proof of Lemma~\ref{dec-assign2}, we have
    \begin{align*}
        \xi_{S}^{(1)}
        & \geq w_{5} + \sum_{3\leq i\leq 5}{t_{i,1}(x)\delta_i} + \sum_{4\leq i\leq 5}{t_{i,2}(x)(\delta_i+\delta_{i-1})}         \\
        & = w_{5} + \sum_{3\leq i\leq 5}{(n_i(x) - 2t_{i,2}(x))\delta_i} + \sum_{4\leq i\leq 5}{t_{i,2}(x)(\delta_i+\delta_{i-1})}\\
        & = w_{5} + \sum_{3\leq i\leq 5}n_i(x)\delta_i -\sum_{3\leq i\leq 5}2t_{i,2}(x)\delta_i + \sum_{4\leq i\leq 5}t_{i,2}(x)(\delta_{i-1}+\delta_i).
    \end{align*}
    As mentioned before, $t_{3,2}(x)=0$ holds, so we have
    \begin{align*}
        \xi_{S}^{(1)}
        &\geq w_{5} + \sum_{3\leq i\leq 5}n_i(x)\delta_i -\sum_{4\leq i\leq 5}2t_{i,2}(x)\delta_i + \sum_{4\leq i\leq 5}t_{i,2}(x)(\delta_{i-1}+\delta_i)\\
        &= w_{5} + \sum_{3\leq i\leq 5}n_i(x)\delta_i + \sum_{4\leq i\leq 5}t_{i,2}(x)(\delta_{i-1}-\delta_i).
    \end{align*}
    Since $\delta_{i-1}-\delta_i=0$ when $i=4$, we have
    \begin{align*}
        \xi_{S_1}^{(1)} \geq w_5 + \sum_{3\leq i\leq 5}n_i(x)\delta_i + t_{5,2}(x)(\delta_4-\delta_5).
    \end{align*}
    Similarly, we can get
    \begin{align*}
        \xi_{S_2}^{(1)} \geq w_5 + \sum_{3\leq i\leq 5}n_i(\overline{x})\delta_i + t_{5,2}(\overline{x})(\delta_4-\delta_5).
    \end{align*}
    Summing $\xi_{S_1}^{(1)}$ and $\xi_{S_2}^{(1)}$ up, we have
    \begin{align*}
        \xi_{S_1}^{(1)}+\xi_{S_2}^{(1)}
        & \geq 2w_5 + \sum_{3\leq i\leq 5}(n_i(x)+n_i(\overline{x}))\delta_i + (t_{5,2}(x)+t_{5,2}(\overline{x}))(\delta_4-\delta_5).
    \end{align*}
\end{proof}

The following lemma is a corollary of Lemma~\ref{dec-sum-5d3c-t}, which will also be used in our analysis to simplify the arguments.
\begin{lemma}\label{dec-sum-5d3c}
    Assume that $\mathcal{F}$ is a reduced CNF-formula of $d=5$. Let $S_1=\{x\}$ and $S_2=\{\overline{x}\}$, where $x$ is a $(2,3)/(3,2)$-literal in $\mathcal{F}$. If all clauses containing $x$ or $\overline{x}$ are $3^+$-clauses, $\sum_{3\leq i\leq 5}(n_i(x)+n_i(\overline{x}))\geq g$, and $\sum_{3\leq i\leq 4}(n_i(x)+n_i(\overline{x}))\geq h$, then it holds that
    \begin{align*}
        \xi_{S_1}^{(1)}+\xi_{S_2}^{(1)}
        & \geq 2w_5 + g\delta_5 + h(w_3-\delta_5) + (t_{5,2}(x)+t_{5,2}(\overline{x}))(\delta_4-\delta_5).
    \end{align*}
\end{lemma}
\begin{proof}
    By Lemma~\ref{dec-sum-5d3c-t} we have
    \begin{align*}
        \xi_{S_1}^{(1)}+\xi_{S_2}^{(1)}
        & \geq 2w_5 + \sum_{3\leq i\leq 5}(n_i(x)+n_i(\overline{x}))\delta_i + (t_{5,2}(x)+t_{5,2}(\overline{x}))(\delta_4-\delta_5) \\
        & =    2w_5 + \sum_{3\leq i\leq 4}(n_i(x)+n_i(\overline{x}))\delta_i + (n_5(x)+n_5(\overline{x}))\delta_5\\
        & \quad + (t_{5,2}(x)+t_{5,2}(\overline{x}))(\delta_4-\delta_5).
    \end{align*}
    Let $g'=\sum_{3\leq i\leq 5}(n_i(x)+n_i(\overline{x}))\geq g$ and $h'=\sum_{3\leq i\leq 4}(n_i(x)+n_i(\overline{x}))\geq h$.
    Since $\delta_3=\delta_4=w_3$ and $n_5(x)+n_5(\overline{x})= g'-h'$, we have
    \begin{align*}
        \xi_{S_1}^{(1)}+\xi_{S_2}^{(1)}
        & \geq 2w_5 + h'w_3 + (g'-h')\delta_5 + (t_{5,2}(x)+t_{5,2}(\overline{x}))(\delta_4-\delta_5)\\
        & =    2w_5 + g'\delta_5 + h'(w_3-\delta_5) + (t_{5,2}(x)+t_{5,2}(\overline{x}))(\delta_4-\delta_5).
    \end{align*}
    Note that $\delta_5>0$ and $w_3-\delta_5>0$, so it holds that
    \begin{align*}
        \xi_{S_1}^{(1)}+\xi_{S_2}^{(1)} \geq 2w_5 + g\delta_5 + h(w_3-\delta_5) + (t_{5,2}(x)+t_{5,2}(\overline{x}))(\delta_4-\delta_5).
    \end{align*}
\end{proof}

\section{Step Analysis}
Equipped with the above lower bounds, we are ready to analyze the branching vector of each step in the algorithm.

\subsection{Step~2}
\textbf{Step 2.} If $\mathcal{F}$ is not a reduced CNF-formula, iteratively apply the reduction rules to reduce it.

In this step, we only apply reduction rules to reduce the formula.
However, it is still important to show that the measure will never increase when applying reduction rules, and reduction operations use only polynomial time.

\begin{lemma}\label{reduction-measure}
    For any CNF-formula $\mathcal{F}$, it holds that
    $$\mu(R(\mathcal{F}))\leq \mu(\mathcal{F}).$$
\end{lemma}
\begin{proof}

    It suffices to verify that each reduction rule would not increase the measure of the formula.

    R-Rules 1-8 simply remove some literals, which would not increase the measure of the formula.

    Next, we consider R-Rule \ref{rule-link} and \ref{rule-back-resol}. Note that now
    R-Rule \ref{rule-pure} and \ref{rule-trivial-resol} are not applicable, and then all variables in $\mathcal{F}$ are $3^+$-variables.
    The reason is below.
    If there is a $1$-variable, we could apply R-Rule \ref{rule-pure}. For a $2$-variable in $\mathcal{F}$, if it is a $(2, 0)$-variable or $(0, 2)$-variable, R-Rule \ref{rule-pure} would be applicable; if it is a $(1, 1)$-variable, R-Rule \ref{rule-trivial-resol} would be applicable.

    For R-Rule \ref{rule-link}, without loss of generalization, we assume the degree of $z_1$ is $i$ and the degree of $z_2$ is $j$ such that $3\leq i\leq j$. After applying this rule, the degree of $z_2$ will become $i+j-2$ since we replace $z_1$ with $\overline{z_2}$ and clause $z_1z_2$ is removed by R-Rule \ref{rule-tauto}. Next, we show that $\mu(R(\mathcal{F}))-\mu(\mathcal{F})=w_{i+j-2}-w_i-w_j\leq 0$ holds.

    By (\ref{weight1}), (\ref{weight2}) and (\ref{eq-w3-d3-d4}), we know that: \\
    if $i=3$, then $$w_{3+j-2}-w_3-w_j=w_{j+1}-w_{j}-w_3=\delta_{j+1}-w_3\leq \delta_4-w_3 = 0;$$
    if $i=4$, then $$w_{4+j-2}-w_4-w_j=w_{j+2}-w_{j}-w_4=\delta_{j+2}+\delta_{j+1}-w_4\leq \delta_6+\delta_5-w_4<0;$$
    if $i\geq 5$, then $$w_{i+j-2}-w_i-w_j=(i+j-2)-i-j=-2<0.$$

    For R-Rule \ref{rule-back-resol},  two clauses $CD_1$ and $CD_2$ in $\mathcal{F}$ are replaced with three clauses $xC, \overline{x}D_1$ and $\overline{x}D_2$. We introduce a $3$-variable $x$ and also decrease the degree of each literal in $C$ by 1. The introduction of $x$ increases the measure of the formula by $w_3$. On the other hand, since all variables are $3^+$-variable in $\mathcal{F}$ and $|C|\geq 2$, the removing of clause $C$ decreases the measure at least $2\min\{\delta_i|i\geq 3\}=2$. Since $w_3<2$, we know that $$\mu(R(\mathcal{F}))-\mu(\mathcal{F})=w_3-2\min\{\delta_i|i\geq 3\} = w_3-2 < 0.$$
\end{proof}

\begin{lemma}\label{Rpolytime}
    For any CNF-formula $\mathcal{F}$,
    we can apply the reduction rules in polynomial time to transfer it to $R(\mathcal{F})$.
\end{lemma}
\begin{proof}
    Each application of any one of the first eight reduction rules removes some literal from the formula, and then it decreases $L(\mathcal{F})$ at least by 1.
    For R-Rule \ref{rule-link}, after replacing $z_1$ with $\overline{z_2}$, the clause $z_1z_2$ will be removed by R-Rule \ref{rule-tauto}. So each application of it decreases $L(\mathcal{F})$ by at least 2.

    It is easy to see that each application of R-Rule \ref{rule-back-resol} increases $L(\mathcal{F})$ by at most 1.
    In the proof of Lemma \ref{reduction-measure}, we have shown that the measure $\mu(\mathcal{F})$ decreases at least by $2-w_3$ in this step. Since $w_3$ is a constant less than $2$, let $w_3=2-\epsilon$ for a constant $\epsilon$, then $\mu(\mathcal{F})$ decreases by at least $\epsilon$ after applying R-Rule \ref{rule-back-resol}.

    In order to make the proof clear, we define a new measure $M(\mathcal{F})=L(\mathcal{F})+2\mu(\mathcal{F})/\epsilon$.
    It is easy to see that  $M(\mathcal{F})$ is bounded by a polynomial of $L(\mathcal{F})$.
    For R-Rule 1-9, it decreases $M(\mathcal{F})$ by at least 1. For R-Rule \ref{rule-back-resol}, it increases $M(\mathcal{F})$ by at most 1. Assume the resulting CNF-formula after applying R-Rule \ref{rule-back-resol} is $\mathcal{F}'$. It holds that $L(\mathcal{F})-L(\mathcal{F}')\geq -1$ and $\mu(\mathcal{F})-\mu(\mathcal{F}')\geq \epsilon$. Sod we have
    $$M(\mathcal{F})-M(\mathcal{F}')=L(\mathcal{F})-L(\mathcal{F}')+2(\mu(\mathcal{F})-\mu(\mathcal{F}'))/\epsilon \geq -1+2=1.$$
    Since each reduction rule decreases $M(\mathcal{F})$ at least by 1, it must stop in polynomial time if we iteratively apply the reduction rules.
\end{proof}

\subsection{Step~3}\label{step-3}
\textbf{Step 3.} If the degree of $\mathcal{F}$ is at least $6$, select a variable $x$ with the maximum degree and return SAT($\mathcal{F}_{x=1}$)$\vee$SAT($\mathcal{F}_{x=0}$).

In this step, we branch on a variable $x$ of degree at least 6. The two sub-branches are:  $S_1=\{x\}$;  $S_2=\{\overline{x}\}$.
We have the following result:
\begin{lemma}
    The branching vector generated by Step~3 is covered by
    \begin{equation}\label{bvec-step3}
        [w_6+\delta_6, w_6+11\delta_6].
    \end{equation}
\end{lemma}
\begin{proof}
    Since R-Rule~\ref{rule-pure} is not applicable, both $x$ and $\overline{x}$ are $(1^+,1^+)$-literals in $\mathcal{F}$.
    By Lemma~\ref{coro-assign2} and $\delta_d=\delta_6$ for $d\geq 6$, we get that
    \begin{equation*}
        \Delta_{S_1}\geq \xi_{S_1}^{(1)} \geq w_d + j\delta_d \geq w_6 + \delta_6
    \end{equation*}
    since $x$ is a $(j, d-j)$-literal with $j\geq 1$. Similarly, we have $\Delta_{S_2}\geq w_6+\delta_6$.

    By Lemma~\ref{dec-sum}, $w_3>\delta_d$ and $w_i>2\delta_d$ for $4\leq i\leq d$, we have that
    \begin{align*}
        \Delta_{S_1}+\Delta_{S_2}
        \geq \Delta_{S_1}^*+\Delta_{S_2}^*
         & \geq 2w_d + 2d\delta_d + (n'_3(x)+n'_3(\overline{x})) (2w_3-2\delta_d)  \\
         & \quad +\sum_{4\leq i\leq d}{(n'_i(x)+n'_i(\overline{x}))(w_i-2\delta_d)} \\
         & \geq 2w_6 + 12\delta_d.
    \end{align*}
    Since $\delta_d=\delta_6$ for $d\geq 6$, we obtain
    \[
        \Delta_{S_1}+\Delta_{S_2} \geq 2w_6 + 12\delta_6.
    \]

    As $\min(\Delta_{S_1},\Delta_{S_2})\geq w_6+\delta_6$ and $\Delta_{S_1}+\Delta_{S_2}\geq 2w_6+12\delta_6$, by Lemma~\ref{lemma-bvec}, we know that the branching vector generated by this step is covered by
    \begin{equation*}
        [w_6+\delta_6, w_6+11\delta_6].
    \end{equation*}
\end{proof}

\subsection{Step~4}\label{step-4}
\textbf{Step 4.} If there is a $(1,4)$-literal $x$ (assume that $xC$ is the unique clause containing $x$), return SAT($\mathcal{F}_{x=1\And C=0}$)$\vee$SAT($\mathcal{F}_{x=0}$).

After Step~3, the degree of $\mathcal{F}$ is at most $5$.
In this step, the algorithm will branch on a $(1,4)$-literal $x$. The two sub-branches are: $S_1=\{x\}\cup \overline{C}$; $S_2=\{\overline{x}\}$.
We have the following result:
\begin{lemma}
    The branching vector generated by Step~4 is covered by
    \begin{equation}\label{bvec-step4}
        [w_5+2w_3, w_5+w_3+7\delta_5].
    \end{equation}
\end{lemma}
\begin{proof}
    By Lemma~\ref{dec-1ltr}, we get
    \begin{align*}
        \Delta_{S_1} + \Delta_{S_2} \geq 2w_5+3w_3+7\delta_5 \text{~~and~~}
        \min(\Delta_{S_1}, \Delta_{S_2})\geq w_5+2w_3.
    \end{align*}

    By (\ref{ineq-2d5>w3}), we have $2w_3\leq 4\delta_5$. Thus $\min(\Delta_{S_1}, \Delta_{S_2}) \geq w_5+2w_3$. By Lemma~\ref{lemma-bvec}, we know that the branching vector of this step is covered by
    \begin{equation*}
        [w_5+2w_3, w_5+w_3+7\delta_5].
    \end{equation*}
\end{proof}

\subsection{Step~5}\label{step-5}
\textbf{Step 5.} If there is a $5$-literal $x$ such that at least two $2$-clauses contain $x$ or $\overline{x}$, return SAT($\mathcal{F}_{x=1}$)$\vee$SAT($\mathcal{F}_{x=0}$).

Note that after Step~4, $x$ is either a $(2, 3)$-literal or $(3, 2)$-literal.
In this step, the two sub-branches are: $S_1=\{x\}$; $S_2=\{\overline{x}\}$.
We have the following result:
\begin{lemma}
    The branching vector generated by Step~5 is covered by
    \begin{equation}\label{bvec-step5}
        [w_5+2\delta_5, w_5+4w_3+4\delta_5].
    \end{equation}
\end{lemma}
\begin{proof}
    Since $x$ is a $(j,5-j)$-literal with $2\leq j\leq 3$, by Lemma \ref{coro-assign2}, we have
    \begin{equation*}
        \Delta_{S_1}\geq \xi_{S_1}^{(1)} \geq w_5+2\delta_5.
    \end{equation*}
    Similarly, we can get $\Delta_{S_2}\geq w_5+2\delta_5$.

    By Lemma \ref{dec-sum}, $w_5\geq w_4\geq 2w_3$, and $\sum_{3\leq i\leq 5}n'_i(x)+n'_i(\overline{x})\geq 2$, we have
    \begin{align*}
        \Delta_{S_1}+\Delta_{S_2}
        \geq \Delta_{S_1}^{*}+\Delta_{S_2}^{*}
        & \geq 2w_5 + 2\cdot 5\delta_5 + (n'_3(x)+n'_3(\overline{x}))(2w_3-2\delta_5)                        \\
        & \quad + \sum_{4\leq i\leq 5}{(n'_i(x)+n'_i(\overline{x}))(w_i-2\delta_5)}.
    \end{align*}
    Note that $w_5\geq w_4=2w_3$ by (\ref{weight1}). We can get $w_i-2\delta_5\geq 2w_3-2\delta_5$ for $3\leq i\leq 4$. So we have
    \begin{align*}
        \Delta_{S_1}+\Delta_{S_2}\geq 2w_5 + 10\delta_5 + \sum_{3\leq i\leq 5}{(n'_i(x)+n'_i(\overline{x}))(2w_3-2\delta_5)}.
    \end{align*}
    Since there are at least two $2$-clauses containing literal $x$ or $\overline{x}$, it holds that $\sum_{3\leq i\leq d}(n'_i(x)+n'_i(\overline{x}))\geq 2$.
    We further obtain
    \begin{align*}
        \Delta_{S_1}+\Delta_{S_2}
        & \geq 2w_5 + 10\delta_5 + \sum_{3\leq i\leq 5}{(n'_i(x)+n'_i(\overline{x}))(2w_3-2\delta_5)} \\
        & =    2w_5 + 10\delta_5 + (\sum_{3\leq i\leq 5}{(n'_i(x)+n'_i(\overline{x}))})(2w_3-2\delta_5)  \\
        & \geq 2w_5 + 10\delta_5 + 2(2w_3-2\delta_5)                                                  \\
        & =    2w_5 + 4w_3 + 6\delta_5
    \end{align*}

    Since $\min(\Delta_{S_1}, \Delta_{S_2})\geq w_5+2\delta_5$, by Lemma~\ref{lemma-bvec}, the branching vector of this step is covered by
    \begin{equation*}
        [w_5+2\delta_5, w_5+4w_3+4\delta_5].
    \end{equation*}
\end{proof}

\subsection{Step~6}\label{step-6}

\textbf{Step 6.} If there are two $5$-literals $x$ and $y$ contained in one $2$-clause $xy$, return SAT($\mathcal{F}_{x=1}$)$\vee$SAT($\mathcal{F}_{x=0}$).

After Step~5, a $5$-variable can be contained in at most one $2$-clause.
In this step, if there is a $2$-clause $xy$ containing two $5$-variables, we pick one of the 5-variables, say $x$, and branch on it. The two sub-branches are: $S_1=\{x\}$; $S_2=\{\overline{x}\}$.

This case will not be the bottleneck case in our algorithm. We will show that after branching some bottleneck cases, we can always get this step. This implies we can use the shift technique here. We will save a shift $\sigma > 0$ from the branching vector of this step that will be included in some bad branching vectors. The value of $\sigma$ will be decided later. We have the following result for this step.
\begin{lemma}
    The branching vector generated by Step~6 is covered by
    \begin{equation} \label{bvec-step6}
        [w_5+3\delta_5-\sigma, 2w_5+2w_3+3\delta_5-\sigma]
    \end{equation}
\end{lemma}
\begin{proof}
Assume that $x$ is a ($j, 5-j$)-literal, where $j=2$ or 3. Since $x$ is contained in only one 2-clause and there is no 1-clause now, we have that
     $\sum_{i\geq 3}n_i(x)\geq 2j-1$.
    By Lemma~\ref{dec-assign2}, we get
    \begin{equation*}
        \Delta_{S_1}\geq \xi_{S_1}^{(1)} \geq w_5 + (\sum_{3\leq i\leq 5}n_i(x))\delta_5\geq w_5 + (2j-1)\delta_5.
    \end{equation*}
    Next, we analyze $\Delta_{S_2}$.
    By Lemma~\ref{dec-assign2}, we first get
    \begin{equation*}
    \xi_{S_2}^{(1)} \geq w_5 + \sum_{3\leq i\leq 5}n_i(\overline{x})\delta_i.
    \end{equation*}
    We look at $\mathcal{F}_{S_2=1}$, which is the formula after assigning 1 to $\overline{x}$ in $\mathcal{F}$.
    By Lemma~\ref{prop2}, we know that in $\mathcal{F}_{S_2=1}$, $var(y)$ is a variable of degree at least 4 and there is a 1-clause $\{y\}$.
    Let $P = \{z: z\in N(y, \mathcal{F}) \text{~and~} z\cap\{x, \overline{x}\}=\emptyset\}$ and $Q = \{z: z\in N(y, \mathcal{F}_{S_2=1}) \text{~and~} deg(z)\geq 3\}$.
    By applying R-Rule~\ref{rule-pure}, we will assign 1 to $y$ and remove the literals in $N(y, \mathcal{F}_{S_2=1})$, which will further reduce the measure.
    Thus, we have $\xi_{S_2}^{(2)}\geq w_4+|Q|\delta_5$ since $\delta_5\leq \delta_4\leq \delta_3$.
    We get that
    \begin{equation*}
        \Delta_{S_2}
        \geq \xi_{S_2}^{(1)} + \xi_{S_2}^{(2)}
        \geq w_5 +\sum_{3\leq i\leq 5}n_i(\overline{x})\delta_i +  w_4+|Q|\delta_5.
    \end{equation*}
    If there is a literal $z$ such that $z\in P$ and $z\notin Q$, then $z$ must be a neighbor of $\overline{x}$ in $\mathcal{F}$ with a degree of at most $4$ by Lemma~\ref{prop2}.
    In other words, it holds that $n_3(\overline{x})+n_4(\overline{x})\geq |P|-|Q|$.
    Since $\overline{x}$ is a ($5-j, j$)-literal not contained any 2-clause or 1-clause, we have that $\sum_{3\leq i\leq 5}n_i(\overline{x})\geq 2(5-j)=10-2j$, which implies $n_5(\overline{x})\geq 10-2j-(n_3(\overline{x})+n_4(\overline{x}))$.
    With $n_3(\overline{x})+n_4(\overline{x})\geq |P|-|Q|$ and $\delta_3=\delta_4=w_3$, we further get
    \begin{align*}
        \Delta_{S_2}
        & \geq w_5 +\sum_{3\leq i\leq 5}n_i(\overline{x})\delta_i +  w_4+|Q|\delta_5                \\
        & \geq 2w_5 + (n_3(\overline{x})+n_4(\overline{x}))w_3 + (|Q|-1+n_5(\overline{x}))\delta_5  \\
        & \geq 2w_5 + (n_3(\overline{x})+n_4(\overline{x}))w_3 + (|Q|-1+10-2j-(n_3(\overline{x})+n_4(\overline{x})))\delta_5 \\
        & =    2w_5 + (n_3(\overline{x})+n_4(\overline{x}))(w_3-\delta_5) + (|Q|+9-2j)\delta_5 \\
        & \geq 2w_5 + (|P|-|Q|)(w_3-\delta_5) + (|Q|+9-2j)\delta_5 \\
        & =    2w_5 + (9-2j)\delta_5 + |P|(w_3-\delta_5) + |Q|(2\delta_5-w_3).
    \end{align*}
    Note that in $\mathcal{F}$, literal $y$ is also a (2,3)/(3,2)-literal contained in exactly one 2-clause $xy$ (since Step 5 has been applied).
    There is another clause containing $y$ and two different literals $z_1$ and $z_2$, where $\{z_1,z_2\}\cap \{x,\overline{x}\}=\emptyset$ by Lemma~\ref{prop2}.
    So $|P|\geq 2$ holds. With $2\delta_5>w_3$, we get
    \begin{align*}
        \Delta_{S_2} \geq 2w_5+(9-2j)\delta_5+2(w_3-\delta_5)\geq 2w_5+2w_3+(7-2j)\delta_5.
    \end{align*}
    It is easy to see that the case of $j=2$ covers the case of $j=3$. For $j=2$, we get a branching vector
     \begin{equation*}[\Delta_{S_1}, \Delta_{S_2}]=[w_5+3\delta_5, 2w_5+2w_3+3\delta_5].
     \end{equation*}

    After saving a shift $\sigma$ from each branch, we get the following branching vector
    \begin{equation*}
        [w_5+3\delta_5-\sigma, 2w_5+2w_3+3\delta_5-\sigma].
    \end{equation*}

\end{proof}

\subsection{Step~7}\label{step-7}
\textbf{Step 7.} If there is a $5$-literal $x$ contained in a $2$-clause, return SAT$(\mathcal{F}_{x=1})\vee$ SAT$(\mathcal{F}_{x=0})$.

In this Step, if there is a $5$-literal $x$ contained in a $2$-clause $xy$, then $y$ must be a $4^-$-variable.
We branch on $x$. The two sub-branches are: $S_1=\{x\}$; $S_2=\{\overline{x}\}$.
We have the following result:
\begin{lemma}
    The branching vector generated by Step~7 is covered by
    \begin{equation}\label{bvec-step7}
        [w_5+w_3+2\delta_5, w_5+w_3+6\delta_5].
    \end{equation}
\end{lemma}
\begin{proof}

    Note that there is at most one $2$-clause containing $x$ or $\overline{x}$ after Step~5. All clauses containing $x$ are $3^+$-clauses except clause $xy$ and all clauses containing $\overline{x}$ are $3^+$-clauses.
    So it holds that
    \[
        \sum_{3\leq i\leq 5}n_i(x) \geq 3 \text{~and~} \sum_{3\leq i\leq 5}n_i(\overline{x}) \geq 4
    \]
    since both $x$ and $\overline{x}$ are $(2,3)/(3,2)$-literals.

    As $y$ is a $4^-$-variable, $n_3(x)+n_4(x)\geq 1$ holds. With Lemma~\ref{dec-assign2} and $w_3=\delta_3=\delta_4$, we have
    \begin{align*}
        \Delta_{S_1} \geq \xi_{S_1}^{(1)}
        & \geq w_5 + \sum_{3\leq i\leq 5}n_i(x)\delta_i\\
        &=    w_5 + (n_3(x)+n_4(x))w_3 + n_5(x)\delta_5\\
        &\geq w_5 + w_3 + 2\delta_5.
    \end{align*}
    For $\Delta_{S_2}$, by Lemma~\ref{dec-assign2} again and $\delta_3=\delta_4\geq \delta_5$, we get
    \begin{align*}
        \Delta_{S_2}\geq \xi_{S_2}^{(1)}
        \geq w_5 + \sum_{3\leq i\leq 5}n_i(\overline{x})\delta_i
        \geq w_5 + (\sum_{3\leq i\leq 5}n_i(\overline{x}))\delta_5
        \geq w_5 + 4\delta_5.
    \end{align*}
    By the condition of this step, we have
    \[
        \sum_{3\leq i\leq 4}n_i'(x)=1 \text{~and~} n_5'(x)+n_5'(\overline{x})=0.
    \]
    With Lemma~\ref{dec-sum} and $w_4=2w_3$, we have
    \begin{align*}
        \Delta_{S_1}+\Delta_{S_2}
        & \geq \Delta_{S_1}^*+\Delta_{S_2}^*                                         \\
        & \geq 2w_5 + 2\cdot 5\delta_5 + (n'_3(x)+n'_3(\overline{x}))(2w_3-2\delta_5)      \\
        & \quad + \sum_{4\leq i\leq 5}{(n'_i(x)+n'_i(\overline{x}))(w_i-2\delta_5)}  \\
        & \geq 2w_5 + 10\delta_5 + (2w_3-2\delta_5)                                  \\
        & = 2w_5 + 2w_3 + 8\delta_5.
    \end{align*}

    Since $\min(\Delta_{S_1}, \Delta_{S_2})\geq w_5 + 3\delta_5$ and $\Delta_{S_1} + \Delta_{S_2}\geq 2w_5 + 2w_3 + 8\delta_5$, by Lemma~\ref{lemma-bvec}, the branching vector of this step is covered by
    \begin{equation*}
        [w_5+w_3+2\delta_5, w_5+w_3+6\delta_5].
    \end{equation*}
\end{proof}

\begin{lemma}\label{prop-aft-s7}
    After Step~7, if we branch on a $5$-literal $x$ and the two sub-branches are $S_1=\{x\}$ and $S_2=\{\overline{x}\}$, then it holds that:
    \begin{align*}
        \sum_{3\leq i\leq 5}(n_i(x)+n_i(\overline{x}))\geq 10 \text{~and~}
        \min(\Delta_{S_1}, \Delta_{S_2})\geq w_5 + 4\delta_5.
    \end{align*}
\end{lemma}
\begin{proof}
    Note that after Step~7, all clauses containing $x$ or $\overline{x}$ are $3^+$-clauses.
    So it holds that
    \[
        \sum_{3\leq i\leq 5}(n_i(x)+n_i(\overline{x}))\geq 5\cdot 2 = 10.
    \]
    We also have $\sum_{3\leq i\leq 5}n_i(x)\geq 4$ since $x$ is a $(2,3)/(3,2)$-literal.
    By Lemma~\ref{dec-assign2} and $\delta_3=\delta_4\geq \delta_5$, we have
    \begin{align*}
        \Delta_{S_1}\geq \xi_{S}^{(1)}
        \geq w_5 + \sum_{3\leq i\leq 5}n_i(x)\delta_i
        \geq w_5 + (\sum_{3\leq i\leq 5}n_i(x))\delta_5
        \geq w_5 + 4\delta_5.
    \end{align*}
    We can also get $\Delta_{S_2}\geq w_5 + 4\delta_5$ in a similar way.
    Thus the lemma holds.
\end{proof}
The above lemma shows some properties after Step~7.
We will use it in the next several subsections, and we focus on analyzing the lower bound of $\Delta_{S_1}+\Delta_{S_2}$ to get the branching vectors of Step~8-12.

\subsection{Step~8}\label{step-8}
\textbf{Step 8.} If there is a $5$-literal $x$ such that $N(x, \mathcal{F})$ and $N(\overline{x}, \mathcal{F})$ contain at least two $4^-$-literals, return SAT($\mathcal{F}_{x=1}$)$\vee$SAT($\mathcal{F}_{x=0}$).

After Step~7, all clauses containing $x$ or $\overline{x}$ are $3^+$-clauses.
In this step, we branch on a variable $x$ such that there are at least two literals of $4^-$-variables in $N(x, \mathcal{F})$ and $N(\overline{x}, \mathcal{F})$. The two sub-branches are: $S_1=\{x\}$; $S_2=\{\overline{x}\}$. We have the following result:
\begin{lemma}
    The branching vector generated by Step~8 is covered by
    \begin{equation}\label{bvec-step8}
        [w_5+4\delta_5, w_5+2w_3+4\delta_5].
    \end{equation}
\end{lemma}
\begin{proof}
    By lemma~\ref{prop-aft-s7} and the condition of this case, we have
    \[
        \sum_{3\leq i\leq 5}(n_i(x) + n_i(\overline{x}))\geq 10 \text{~and~} \sum_{3\leq i\leq 4}(n_i(x) + n_i(\overline{x}))\geq 2.
    \]
    By Lemma~\ref{dec-sum-5d3c}, we have
    \begin{align*}
        \Delta_{S_1}+\Delta_{S_2}
        \geq \xi_{S_1}^{(1)}+\xi_{S_2}^{(1)}
        \geq 2w_5+10\delta_5+2(w_3-\delta_5)
        = 2w_5+8\delta_5+2w_3.
    \end{align*}

    As $\min(\Delta_{S_1}, \Delta_{S_2})\geq w_5+4\delta_5$ by Lemma~\ref{prop-aft-s7}, by Lemma~\ref{lemma-bvec} we know that the branching vector of this step is covered by
    \begin{equation*}
        [w_5+4\delta_5, w_5+2w_3+4\delta_5].
    \end{equation*}
\end{proof}

\begin{lemma}\label{reduce-measure}
    Let $\mathcal{F}$ be a reduced CNF-formula. After Step~8, if $\mathcal{F}_{x=1} \neq \mathcal{F}_{x=1}'$ for a $(2, 3)/(3, 2)$-literal $x\in \mathcal{F}$, then it holds that
    \[
        \mu(\mathcal{F}) - \mu(\mathcal{F}_{x=1}')\geq w_3-1.
    \]
    In other words, if we can apply reduction rules on $\mathcal{F}_{x=1}$, the measure would decrease by at least $w_3-1$.
\end{lemma}
\begin{proof}
    R-Rules~1-8 only remove some literals, and so applying any one of them decreases the measure by at least $\delta_5$.

    After Step~7, all clauses containing variable $x$ are $3^+$-clauses, some $2$-clauses would be generated in $\mathcal{F}_{x=1}$, and so R-Rule~\ref{rule-link} may be applicable in $\mathcal{F}_{x=1}$ if R-Rules~1-8 do not apply.
    After Step~8, there is at most one $4^-$-variable in $N(x, \mathcal{F})$ and $N(\overline{x}, \mathcal{F})$, so all $2$-clauses in $\mathcal{F}_{x=1}$ contain at least one $5$-variable.
    Thus applying R-Rule~\ref{rule-link} decreases the measure by at least $\min\{w_5+w_i-w_{5+i-2}|3\leq i\leq 5\}=w_3-1$.

    For R-Rule~\ref{rule-back-resol}, we claim that if R-Rule~1-9 are not applicable on $\mathcal{F}_{x=1}$, then R-Rule~\ref{rule-back-resol} is also not applicable. The reason is as follows.
    If there exists two clauses $CD_1$ and $CD_2$ in $\mathcal{F}_{x=1}$ such that R-Rule~\ref{rule-back-resol} is applicable on $\mathcal{F}_{x=1}$, then there must be two clauses $C'D'_1$ and $C'D'_2$ in $\mathcal{F}$ such that $C\subseteq C'$, $D_1\subseteq D'_1$, and $D_2\subseteq D'_2$ since $\mathcal{F}_{x=1}$ is obtained by removing some literals and clauses from $\mathcal{F}$. This implies we could apply R-Rule~\ref{rule-back-resol} on $\mathcal{F}$, which contradicts the condition that $\mathcal{F}$ is reduced.

    Thus, it holds that either $\mathcal{F}_{x=1}=\mathcal{F}_{x=1}'$ or $\mu(\mathcal{F}) - \mu(\mathcal{F}_{x=1}')\geq \min(\delta_5, w_3-1)=w_3-1$ by the assumption in (\ref{weight2}).
\end{proof}

Recall that after Step~8, any $5$-literal $x$ must be a $(2,3)$/$(3,2)$-literal.
We assume $xC_1, xC_2, \overline{x}D_1, \overline{x}D_2$, and $\overline{x}D_3$ are the five clauses containing $x$ or $\overline{x}$ in the analysis of Step~9 and Step~10.

\subsection{Step~9}\label{step-9}
\textbf{Step 9.} If there exist $5$-literals $y_1$ and $y_2$ such that $y_1\in C_1$, $y_1\in D_1$, $y_2\in C_2$ and $y_2$ or $\overline{y_2}\in D_2$, return SAT($\mathcal{F}_{y_1=1}$)$\vee$SAT($\mathcal{F}_{y_1=0}$).

In this step, the two sub-branches are: $S_1=\{y_1\}$; $S_2=\{\overline{y_1}\}$. We have the following result:
\begin{lemma}
    The branching vector generated by Step~9 is covered by
    \begin{equation}\label{bvec-step9}
        [w_5+4\delta_5, w_5+\delta_4+6\delta_5].
    \end{equation}
\end{lemma}
\begin{proof}
    Note that $x, \overline{x}\in N(y_1)$, and so $var(y_1)$ is a $5$-variable, otherwise we can do Step~8.
    This implies $t_{5,2}(y_1)\geq 1$.

    By Lemma~\ref{dec-sum-5d3c-t}, we have
    \begin{align*}
        \xi_{S_1}^{(1)}+\xi_{S_2}^{(1)}
        & \geq 2w_5 + \sum_{3\leq i\leq 5}(n_i(y_1)+n_i(\overline{y_1}))\delta_i + (t_{5,2}(y_1)+t_{5,2}(\overline{y_1}))(\delta_4-\delta_5)\\
        & \geq 2w_5 + \sum_{3\leq i\leq 5}(n_i(y_1)+n_i(\overline{y_1}))\delta_5 + (t_{5,2}(y_1)+t_{5,2}(\overline{y_1}))(\delta_4-\delta_5)\\
        & \geq 2w_5+10\delta_5+(\delta_4-\delta_5)                                                                                       \\
        & = 2w_5+\delta_4+9\delta_5.
    \end{align*}

    Since R-Rule~\ref{rule-back-resol} is not applicable, literal $y_1$ would not appear in $C_2$, $D_2$ and $D_3$.
    We look at $\mathcal{F}_{S_1=1}$, which is the resulting formula after we assign value 1 to $y_1$.
    In $\mathcal{F}_{S_1=1}$, $x$ becomes a $3$-variable, the three clauses containing $x$ or $\overline{x}$ will be $xC_2, \overline{x}D_2$, and $\overline{x}D_3$, and  $y_2\in C_2$.

    \textbf{Case 1.} If $y_2\in D_2'$, then R-Rule~\ref{rule-1liter-neg} is applicable. Applying any one of R-Rules~1-7 decreases the measure by at least $\delta_5$ since each of them removes at least one literal.

    \textbf{Case 2.} If $\overline{y_2}\in D_2'$, then R-Rule~\ref{rule-trivial-resol} is applicable. Resolution on $x$ decreases the measure by $w_3$.

    So it holds that $\xi_{S_1}^{(2)} + \xi_{S_1}^{(3)}\geq \min(\delta_5, w_3) = \delta_5$ and we have
    \begin{align*}
        \Delta_{S_1}+\Delta_{S_2}
        \geq \xi_{S_1}^{(1)}+\xi_{S_2}^{(1)} +(\xi_{S_1}^{(2)} + \xi_{S_1}^{(3)})
        \geq 2w_5+\delta_4+10\delta_5.
    \end{align*}

    With $\min(\Delta_{S_1}, \Delta_{S_2})\geq w_5+4\delta_5$ by Lemma~\ref{prop-aft-s7}, by Lemma~\ref{lemma-bvec} we have that the branching vector of this step is covered by
    \begin{equation*}
        [w_5+4\delta_5, w_5+\delta_4+6\delta_5].
    \end{equation*}
\end{proof}

\subsection{Step~10}\label{step-10}
\textbf{Step 10.} If there exist $5$-literals $y_1$  and $y_2$ such that $y_1\in C_1$, $\overline{y_1}\in D_1$, $y_2\in C_2$, and $\overline{y_2}\in D_2$, pick a $5$-literal $z\in D_3$ (let $R_5(\mathcal{F}_{z=1})$ denote the resulting formula after we only apply R-Rule~\ref{rule-trivial-resol} on $\mathcal{F}_{z=1}$) and return SAT($R_5(\mathcal{F}_{z=1})$) $\vee$SAT($\mathcal{F}_{z=0}$).

After Step~7, all clauses containing $5$-literals are $3^+$-clauses so $|D_3|\geq 2$.
After Step~8, there is at most one $4^-$-literal in $D_3$.
So there must exist a $5$-litreal $z\in D_3$, we branch on this litreal and the two sub-branches are: $S_1=\{z\}$; $S_2=\{\overline{z}\}$.
Note that we will first apply R-Rule~\ref{rule-trivial-resol} on $\mathcal{F}_{S_1=1}$.
We have the following result:
\begin{lemma}
    The branching vector generated by Step~10 is covered by
    \begin{equation}
        [w_5+4\delta_5, w_5+w_4+6\delta_5].
    \end{equation}
\end{lemma}
\begin{proof}
    By Lemma~\ref{dec-sum-5d3c} with $g\geq 10$ (since $x$ is not contained in any $2$-clause after Step~7), we get
    \begin{align*}
        \xi_{S_1}^{(1)}+\xi_{S_2}^{(1)}
        & \geq 2w_5+10\delta_5.
    \end{align*}

    Note that literal $z$ would not appear in $C_1$ and $C_2$, otherwise Step~9 would be applied.
    We look at $\mathcal{F}_{S_1=1}$, which is the resulting formula after we assign value 1 to $z$.
    In $\mathcal{F}_{S_1=1}$, $x$ becomes a $4$-variable and the four clauses containing $x$ or $\overline{x}$ are $xC_1$, $xC_2$, $\overline{x}D_1$, and $\overline{x}D_2$.
    Since $y_1\in C_1$, $y_2\in C_2$, $\overline{y_1}\in D_1$, and $\overline{y_2}\in D_2$, we can apply R-Rule~\ref{rule-trivial-resol} on $x$.
    This decreases the measure by $w_4$. Thus $\xi_{S_1}^{(3)}\geq w_4$ and we have
    \begin{align*}
        \Delta_{S_1}+\Delta_{S_2}
        \geq \xi_{S_1}^{(1)}+\xi_{S_2}^{(1)}+\xi_{S_1}^{(3)}
        \geq 2w_5+10\delta_5+w_4.
    \end{align*}
    With $\min(\Delta_{S_1}, \Delta_{S_2})\geq w_5+4\delta_5$ by Lemma~\ref{prop-aft-s7}, by Lemma~\ref{lemma-bvec} we have that the branching vector of this step is covered by
    \begin{equation*}
        [w_5+4\delta_5, w_5+w_4+6\delta_5].
    \end{equation*}
\end{proof}

For the sake of presentation, we define an auxiliary $G_x$ for each literal $x$ as follows.

\begin{definition}[Clause-clause incidence graph]
     Let $x$ be a literal in a formula $\mathcal{F}$. Assume the clauses containing $x$ are $xC_1, xC_2, \dots, xC_a$ and the clauses containing $\overline{x}$ are $\overline{x}D_1, \overline{x}D_2, \dots, \overline{x}D_b$.
    A \emph{clause-clause incidence graph} of literal $x$, denoted by $G_x$, is a bipartite graph with bipartition $(X, Y)$ where $X$ is the set of clauses $C_i(1\leq i\leq a)$ and $Y$ is the set of clauses $D_j(1\leq j\leq b)$, and there is an edge between $C_i\in X(1\leq i\leq a)$ and $D_j\in Y(1\leq j\leq b)$ if and only if $C_i$ and $D_j$ contain the literal of the same variable in $\mathcal{F}$.
\end{definition}

\noindent \emph{Example.} Let $\mathcal{F}$ be a formula and $xC$ and $\overline{x}D$ be two clauses in $\mathcal{F}$.
If $y\in C$ and $y$ or $\overline{y}\in D$, then in $G_{x}$, there is an edge between vertex $C$ and vertex $D$. 

\begin{lemma}\label{structure-lemma}
    After Step~10, for any $(2, 3)$-literal $x$, there is no matching of size at least $2$ in $G_x$.
\end{lemma}

\begin{proof}
    If there exists a matching of size $2$ in $G_x$, then there exists two literals $y_1\in C_1$ and $y_2\in C_2$ such that two of $y_1$, $\overline{y_1}$, $y_2$, and $\overline{y_2}$ appear in two of $D_1$, $D_2$, and $D_3$ separately. Thus we would be able to do Step~9 or Step~10.
\end{proof}

\begin{lemma}\label{2cls-structure}
    After Step~10, for any $(2, 3)$-literal $x$, in $G_{x}$ if all vertices have a degree of at most $2$, then there are at least two vertices of degree $0$ in $G_{x}$.
\end{lemma}

\begin{proof}
    If there is at most one vertex of degree $0$, then after deleting the degree-0 vertices, we get a graph with $4$ vertices such that each vertex has a degree of at least $1$ and at most $2$.
    For any case, this graph has a matching of size $2$, which contradicts Lemma~\ref{structure-lemma}.
\end{proof}

\subsection{Step~11}\label{step-11}
\textbf{Step 11.} If there is a $5$-literal $x$ contained in at least one $4^+$-clause, return SAT($\mathcal{F}_{x=1}$)$\vee$ SAT($\mathcal{F}_{x=0}$).

In this step, we branch on a $5$-literal $x$ contained in at least one $4^+$-clause. The two sub-branches are: $S_1=\{x\}$; $S_2=\{\overline{x}\}$.
The branching vector of this step leads to one of the worst branching factors.
But we will prove that the shift $\sigma > 0$ saved in Step~6 (Section \ref{step-6}) can be used in this step to get an improvement.
We have the following result:
\begin{lemma}
    The branching vector generated by Step~11 is covered by
    \begin{equation}
        [w_5+4\delta_5, w_5+w_3+6\delta_5]\text{~or~}[w_5+4\delta_5, w_5+7\delta_5 + \sigma].
    \end{equation}
\end{lemma}
\begin{proof}
    Let $m_4\geq 1$ be the number of $4^+$-clauses containing literal $x$ or $\overline{x}$, then the number of $3$-clauses containing $x$ or $\overline{x}$ is $5-m_4$. We have
    \begin{align*}
        \sum_{3\leq i\leq 5}(n_i(x)+n'_i(x))\geq 2(5-m_4)+3m_4\geq 10+m_4.
    \end{align*}
    Next, we consider several cases. By Lemma~\ref{prop-aft-s7}, we have $\min(\Delta_{S_1}, \Delta_{S_2})\geq w_5 + 4\delta_5$ for all the following cases.
    So we focus on analyzing $\Delta_{S_1}+\Delta_{S_2}$.

    \textbf{Case 1.} There is a variable $y$ such that both $y$ and $\overline{y}$ appear in $N(x, \mathcal{F})$ or $N(\overline{x}, \mathcal{F})$.

    Note that after Step~8, there is at most one $4^-$-literal in $N(x, \mathcal{F})$ or $N(\overline{x}, \mathcal{F})$, and so we have $t_{5, 2}(x)+t_{5, 2}(\overline{x})\geq 1$.
    By Lemma~\ref{dec-sum-5d3c} with $g=10+m_4$ and $h=1$, we have
    \begin{align*}
        \xi_{S_1}^{(1)}+\xi_{S_2}^{(1)}
        \geq 2w_5+(10+m_4)\delta_5 + 1\cdot(\delta_4-\delta_5)
        \geq    2w_5+w_3+10\delta_5.
    \end{align*}
    By Lemma~\ref{lemma-bvec}, the branching vector of this case is covered by
    \[
        [w_5+4\delta_5, w_5+w_3+6\delta_5].
    \]

    \textbf{Case 2.} There are $4^-$-literals in $N(x, \mathcal{F})$ or $N(\overline{x}, \mathcal{F})$, i.e., $\sum_{3\leq i\leq 4}(n_i(x) + n_i(\overline{x}))\geq 1$. By Lemma~\ref{dec-sum-5d3c} with $g=10+m_4$ and $h=1$, we have
    \begin{align*}
        \Delta_{S_1}+\Delta_{S_2}
        \geq \xi_{S_1}^{(1)}+\xi_{S_2}^{(1)}
         & \geq 2w_5 + (10+m_4)\delta_5 + 1\cdot(w_3-\delta_5)  \\
         & \geq 2w_5 + w_3 + 10\delta_5.
    \end{align*}
    By Lemma~\ref{lemma-bvec}, the branching vector of this case is covered by
    \[
        [w_5+4\delta_5, w_5+w_3+6\delta_5].
    \]

    \textbf{Case 3.} There are no $4^-$-literals in $N(x, \mathcal{F})$ or $N(\overline{x}, \mathcal{F})$ and no variables $y$ such that both $y$ and $\overline{y}$ appear in $N(x, \mathcal{F})$ or $N(\overline{x}, \mathcal{F})$, i.e., $\sum_{3\leq i\leq 4}(n_i(x) + n_i(\overline{x}))= 0$ and $t_{5, 2}(x)+t_{5, 2}(\overline{x})= 0$.

    \textbf{Case 3.1.} There are at least two $4^+$-clauses containing $x$ or $\overline{x}$, i.e., $m_4\geq 2$. By Lemma~\ref{dec-sum-5d3c} with $g=10+m_4$, we have
    \begin{align*}
        \Delta_{S_1}+\Delta_{S_2}
        \geq \xi_{S_1}^{(1)}+\xi_{S_2}^{(1)}
         \geq 2w_5+(10+m_4)\delta_5
         \geq 2w_5+12\delta_5.
    \end{align*}
    Since $w_3<2\delta_5$, the branching vector of this case is covered by that of Case 1.

    \textbf{Case 3.2.} There is only one $4^+$-clause containing variable $x$, i.e., $m_4=1$. Similar to case 2.1, by Lemma~\ref{dec-sum-5d3c} with $g=10+m_4$, we have
    \begin{align*}
        \xi_{S_1}^{(1)}+\xi_{S_2}^{(1)} \geq 2w_5+(10+m_4)\delta_5 = 2w_5+11\delta_5.
    \end{align*}

    Next, we consider $\mathcal{F}_{x=1}'$ and $\mathcal{F}_{x=0}'$.

    \textbf{Case 3.2.1.} Some reduction rules can be applied on $\mathcal{F}_{x=1}$ or $\mathcal{F}_{x=0}$.

    By Lemma~\ref{reduce-measure} the total measure will further decreases by at least $w_3-1$ and we have $\Delta_{S_1}+\Delta_{S_2} \geq  2w_5+11\delta_5 + w_3-1$.
    Since $\delta_5\geq 1$, the branching vector of this case is covered by that of Case 1.

    \textbf{Case 3.2.2.} $\mathcal{F}_{x=1}'=\mathcal{F}_{x=1}$ and $\mathcal{F}_{x=0}'=\mathcal{F}_{x=0}$.

    We show that we can apply Step~6 (Section \ref{step-6}) on either $\mathcal{F}_{x=1}'$ or $\mathcal{F}_{x=0}'$ to use the saved shift $\sigma$.
    Look at $G_{x}$ and consider the following two cases.

    (1) There is a vertex of degree at least 3 in $G_{x}$.
    Recall that the five clauses containing literal $x$ are $xC_1, xC_2, \overline{x}D_1, \overline{x}D_2$, and $\overline{x}D_3$.
    If there is a vertex in $G_{x}$ with degree at least $3$, the corresponding clause of this vertex must be $C_1$ or $C_2$, w.l.o.g, let us assume that the clause is $C_1$.
    By the condition of Case 3 and $\mathcal{F}$ being reduced, for a literal $y\in C_1$, at most one of $y$ and $\overline{y}$ will appear in $D_1,D_2,D_3$.
    So if $C_1$ is a degree-3 vertex in $G_{x}$, it must contain at least three different literals.
    By the condition of Case 3.2, we know $xC_1$ is the unique $4^+$-clause containing variable $x$, and so $|C_2|=2$.
    By Lemma~\ref{structure-lemma}, in $G_x$ the degree of vertex $C_2$ is $0$.
    There are no $4^-$-literals in $C_2$ by the condition of Case 3. Thus, in $\mathcal{F}_{x=0}'$, clause $C_2$ contains two $5$-variables, and so we can apply Step~6 on $\mathcal{F}_{x=0}'$.

    (2) All vertices in $G_{x}$ have a degree of at most $2$.
    By Lemma~\ref{2cls-structure} there are at least two vertices in $G_{x}$ with degree $0$.
    By the condition of Case 3, there are no $4^-$-literals in those clauses, and  then we will get some $2$-clause containing two $5$-literals in $\mathcal{F}_{x=1}'$ or $\mathcal{F}_{x=0}'$ and we can further apply Step~6 on $\mathcal{F}_{x=1}'$ or $\mathcal{F}_{x=0}'$.

    Thus, we further get
    \begin{align*}
        \Delta_{S_1}+\Delta_{S_2} \geq \xi_{S_1}^{(1)}+\xi_{S_2}^{(1)} + \sigma = 2w_5+11\delta_5 + \sigma.
    \end{align*}

    Note that $\min(\Delta_{S_1}, \Delta_{S_2})\geq w_5+4\delta_5$. By Lemma~\ref{lemma-bvec}, the branching vector of this case is covered by
    \[
        [w_5+4\delta_5, w_5+7\delta_5 + \sigma].
    \]

    In summary, the branching vector is covered by
    \begin{equation*}
        [w_5+4\delta_5, w_5+w_3+6\delta_5]\text{~or~}[w_5+4\delta_5, w_5+7\delta_5 + \sigma].
    \end{equation*}
\end{proof}

\subsection{Step~12}\label{step-12}
\textbf{Step 12.} If there is a clause containing both a $5$-literal $x$ and a $4^-$-literal, return SAT($\mathcal{F}_{x=1}$)$\vee$SAT($\mathcal{F}_{x=0}$).

After Step~11, all clauses containing $x$ or $\overline{x}$ are $3$-clauses.
In this step, we branch on a $5$-literal $x$ such that there is one $4^-$-literal in $N(x, \mathcal{F})$. The two sub-branches are: $S_1=\{x\}$; $S_2=\{\overline{x}\}$.
Similar to Step~11, the shift $\sigma > 0$ saved in Step~6 (Section \ref{step-6}) will be used in this step.
We have the following result:
\begin{lemma}
    The branching vector generated by Step~12 is covered by
    \begin{equation}
        [w_5+4\delta_5, w_5+4\delta_5+2w_3] \text{~or~} [w_5+4\delta_5, w_5+w_3+5\delta_5 + \sigma].
    \end{equation}
\end{lemma}
\begin{proof}
    By the condition of this step, we have $\sum_{3\leq i\leq 4}(n_i(x)+n_i(\overline{x}))\geq 1$.
    Next, we consider two cases. For all of the following cases, we have $\min(\Delta_{S_1}, \Delta_{S_2})\geq w_5 + 4\delta_5$ by Lemma~\ref{prop-aft-s7}. So we focus on $\Delta_{S_1}+\Delta_{S_2}$.

    \textbf{Case 1.} There is a variable $y$ such that both $y$ and $\overline{y}$ appear in $N(x, \mathcal{F})$ or $N(\overline{x}, \mathcal{F})$.

    Note that after Step~8, there is at most one $4^-$-literal in $N(x, \mathcal{F})$ or $N(\overline{x}, \mathcal{F})$, and so we have $t_{5, 2}(x)+t_{5, 2}(\overline{x})\geq 1$. By Lemma~\ref{dec-sum-5d3c} and $\delta_4=w_3$, we have
    \begin{align*}
        \xi_{S_1}^{(1)}+\xi_{S_2}^{(1)}
        \geq 2w_5+10\delta_5 + 1\cdot(w_3-\delta_5) + 1\cdot(\delta_4-\delta_5)
        =    2w_5+8\delta_5+2w_3.
    \end{align*}
    By Lemma~\ref{lemma-bvec}, the branching vector of this case is covered by
    \[
        [w_5+4\delta_5, w_5+4\delta_5+2w_3].
    \]

    \textbf{Case 2.} There are no variables $y$ such that both $y$ and $\overline{y}$ appear in $N(x, \mathcal{F})$ or $N(\overline{x}, \mathcal{F})$.

    By Lemma~\ref{dec-sum-5d3c}, we first have
    \begin{align*}
        \xi_{S_1}^{(1)}+\xi_{S_2}^{(1)}
        \geq 2w_5+10\delta_5 + 1\cdot(w_3-\delta_5)
        =    2w_5+w_3+9\delta_5.
    \end{align*}

    Similar to Step~11, we consider $\mathcal{F}_{x=1}'$ and $\mathcal{F}_{x=0}'$.

    \textbf{Case 2.1.}  Some reduction rules can be applied on $\mathcal{F}_{x=1}$ or $\mathcal{F}_{x=0}$.

    By Lemma~\ref{reduce-measure} the total measure will further decreases by at least $w_3-1$ and we have $\Delta_{S_1}+\Delta_{S_2} \geq  2w_5+9\delta_5 + 2w_3-1$.
    By Lemma~\ref{lemma-bvec}, the branching vector of this case is covered by
    \[
        [w_5+4\delta_5, w_5+5\delta_5+2w_3-1].
    \]
    Since $\delta_5\geq 1$, this branching vector is covered by that of Case 1.

    \textbf{Case 2.2.} $\mathcal{F}_{x=1}'=\mathcal{F}_{x=1}$ and $\mathcal{F}_{x=0}'=\mathcal{F}_{x=0}$.

    We look at the auxiliary graph $G_x$.
    In $G_{x}$, all the vertices have a degree of at most $2$ since all clauses containing $x$ or $\overline{x}$ are $3$-clauses and no variables $y$ such that both $y$ and $\overline{y}$ appear in $N(x, \mathcal{F})$ or $N(\overline{x}, \mathcal{F})$.
    With Lemma~\ref{2cls-structure}, there are at least two vertices with degree $0$ in $G_x$.
    Let the set of corresponding clauses of them be $E$.
    Since there may be at most one clause in $E$ that contains a $4^-$-literal, there must exist a $2$-clause containing two $5$-literals in $\mathcal{F}_{x=1}'$ or $\mathcal{F}_{x=0}'$, and thus we can apply Step~6 on $\mathcal{F}_{x=1}'$ or $\mathcal{F}_{x=0}'$. So we further get
    \begin{align*}
        \Delta_{S_1}+\Delta_{S_2}\geq 2w_5+w_3+9\delta_5+\sigma.
    \end{align*}
    By Lemma~\ref{lemma-bvec}, the branching vector of this case is covered by
    \[
        [w_5+4\delta_5, w_5+w_3+5\delta_5 + \sigma].
    \]

    In summary, the branching vector of this step is covered by
    \begin{equation*}
        [w_5+4\delta_5, w_5+4\delta_5+2w_3] \text{~or~} [w_5+4\delta_5, w_5+w_3+5\delta_5 + \sigma].
    \end{equation*}

\end{proof}

\subsection{Step~13}\label{step-13}
\textbf{Step 13.} If there are still some 5-literals, then $\mathcal{F}=\mathcal{F}_{5} \wedge \mathcal{F}_{\leq4}$, where $\mathcal{F}_{5}$ is a 3-CNF containing only 5-literals and $\mathcal{F}_{\leq 4}$ contains only 3/4-literals.

In this step, the literals of all $5$-variables form a $3$-SAT instance $\mathcal{F}_{5}$. We apply the $O^*(1.3279^n)$-time algorithm in~\cite{DBLP:conf/icalp/Liu18} for 3-SAT to solve our problem, where $n$ is the number of variables in the instance. Since $w_5=5$, we have that $n=\mu(\mathcal{F}_5)/w_5=\mu(\mathcal{F}_5)/5$. So the running time for this part will be  $$O^*(1.3279^{\mu(\mathcal{F}_5)/w_5})=O^*(1.0584^{\mu(\mathcal{F}_{5})}).$$

\subsection{Step~14}\label{step-14}
\textbf{Step 14.} If there is a $(1,3)$-literal $x$ (assume that $xC$ is the unique clause containing $x$), return SAT($\mathcal{F}_{x=1\And C=0}$)$\vee$SAT($\mathcal{F}_{x=0}$).

After Step~13, all literals in $\mathcal{F}$ are $4^-$-literals. In this step, we branch on a $(1,3)$-literal $x$. The two sub-branches are: $S_1=\{x\}$; $S_2=\{\overline{x}\}$. We have the following result:
\begin{lemma}
    The branching vector generated by Step~14 is covered by
    \begin{equation}
        [w_4+2w_3, w_4+6\delta_4].
    \end{equation}
\end{lemma}
\begin{proof}
    By Lemma~\ref{dec-1ltr}, we have
    \begin{align*}
        \Delta_{S_1} + \Delta_{S_2} \geq 2w_4+3w_3+5\delta_4 = 2w_4+8w_3 \text{~~and~~}
        \min(\Delta_{S_1}, \Delta_{S_2})\geq w_4+2w_3.
    \end{align*}
    By Lemma~\ref{lemma-bvec}, the branching vector of this step is covered by
    \begin{equation*}
        [w_4+2w_3, w_4+6\delta_4].
    \end{equation*}
\end{proof}

\subsection{Step~15}\label{step-15}
\textbf{Step 15.} If there is a $(2,2)$-literal $x$, return SAT($\mathcal{F}_{x=1}$)$\vee$SAT($\mathcal{F}_{x=0}$).

In this step, the two sub-branches are: $S_1=\{x\}$; $S_2=\{\overline{x}\}$. We have the following result:
\begin{lemma}
    The branching vector generated by Step~15 is covered by
    \begin{equation}
        [w_4+2\delta_4, w_4+6\delta_4].
    \end{equation}
\end{lemma}
\begin{proof}
    Since both $x$ and $\overline{x}$ are $(2, 2)$-literals, by Lemma~\ref{coro-assign2} we have
    \begin{align*}
        \Delta_{S_1}\geq \xi_{S_1}^{(1)} =w_4+2\delta_4 \text{~and~} \Delta_{S_2}\geq \xi_{S_2}^{(1)}\geq w_4+2\delta_4.
    \end{align*}

    By Lemma~\ref{dec-sum}, we have
    \begin{align*}
        \Delta_{S_1}+\Delta_{S_2}
        \geq \Delta_{S_1}^*+\Delta_{S_2}^*
        & \geq 2w_4 + 2\cdot 4\delta_4 + (n'_3(x)+n'_3(\overline{x}))(2w_3-2\delta_4)       \\
        & \quad + (n_4'(x)+n_4'(\overline{x}))(w_4-2\delta_4)                      \\
        & = 2w_4+8\delta_4.
    \end{align*}
    By Lemma~\ref{lemma-bvec}, we know that the branching vector is covered by
    \begin{equation*}
        [w_4+2\delta_4, w_4+6\delta_4].
    \end{equation*}
\end{proof}

\subsection{Step~16}\label{step-16}
\textbf{Step 16.} Apply the algorithm by Wahlstr{\"{o}}m~\cite{DBLP:conf/sat/Wahlstrom05} to solve the instance.

All variables are $3$-variables in this step. We apply the $O^*(1.1279^n)$-time algorithm by Wahlstr{\"{o}}m~\cite{DBLP:conf/sat/Wahlstrom05} to solve this special case, where $n$ is the number of variables. For this case,
we have that $n=\mu(\mathcal{F})/w_3$. So the running time of this part is
$$O^*((1.1279^{1/w_3})^{\mu(\mathcal{F})}).$$

\section{The Final Result}
Each of the above branching vectors above will generate a constraint in our quasiconvex program to solve the best value for $w_3$, $w_4$, and $\sigma$.
Let $\alpha_i$ denote the branching factor for branching vector ($i$) where $10\leq i\leq 21$.
We want to find the minimum value $\alpha$ such that $\alpha \geq \alpha_i$ and $\alpha \geq  1.1279^{1/w_3}$ (generated by Step~16)
under the assumptions (\ref{weight1}) and (\ref{weight2}).
By solving this quasiconvex program, we get that $\alpha=1.0638$ by letting $w_3=1.94719$, $w_4=2w_3=3.89438$, and $\sigma=0.86108$.
Note that $\alpha= 1.0638$ is greater than $1.0584$, which is the branching factor generated in Step~13. So $1.0638$ is the worst branching factor in the whole algorithm. By (\ref{measure-leq-L}), we get the following result.
\begin{theorem}
    Algorithm \ref{main-algorithm} solves the SAT problem in $O^*(1.0638^L)$ time.
\end{theorem}

We also show the whole weight setting in Table \ref{tb-measure} and the branching vector of each step under the setting in Table \ref{tb-bv}.
\begin{table}[h]
    \begin{center}
        \caption{The weight setting}\label{tb-measure}
        \begin{tabular}{l|l}
            \hline
            $w_1=w_2=0$  ~~~ & ~~~ $\sigma=0.86108$                      \\
            $w_3=1.94719$ ~~~ & ~~~ $\delta_3=1.94719$                    \\
            $w_4=3.89438$ ~~~ & ~~~ $\delta_4=1.94719$                    \\
            $w_5=5$ ~~~                  & ~~~ $\delta_5=1.10562$        \\
            $w_i=i(i\geq 6)$ ~~~         & ~~~ $\delta_i=1(i\geq 6)$    \\
            \hline
        \end{tabular}
    \end{center}
\end{table}

\begin{table}[h]
    \renewcommand\arraystretch{1.2}
    \begin{center}
        \caption{The branching vector and factor for each step}\label{tb-bv}
        \begin{tabular}{l|c|l}
            \hline
            \textbf{Steps} & \textbf{Branching vectors}                             & \textbf{Factors}           \\
            \hline
            Step 3~~(Section~\ref{step-3}) & $[w_6+\delta_6, w_6+11\delta_6]$                       & $1.0636$                   \\
            \hline
            Step 4~~(Section~\ref{step-4}) & $[w_5+2w_3, w_5+w_3+7\delta_5]$                        & $1.0620$                        \\
            \hline
            Step 5~~(Section~\ref{step-5}) & $[w_5+2\delta_5, w_5+4w_3+4\delta_5]$                  & $1.0624$                   \\
            \hline
            Step 6~~(Section~\ref{step-6}) &$[w_5+3\delta_5-\sigma, 2w_5+2w_3+3\delta_5-\sigma]$    & $1.0633$                        \\
            \hline
            Step 7~~(Section~\ref{step-7}) & $[w_5+w_3+2\delta_5, w_5+w_3+6\delta_5]$               & $1.0638$~*                   \\
            \hline
            Step 8~~(Section~\ref{step-8}) & $[w_5+4\delta_5, w_5+2w_3+4\delta_5]$                  & $1.0636$                   \\
            \hline
            Step 9~~(Section~\ref{step-9}) & $[w_5+4\delta_5, w_5+\delta_4+6\delta_5]$                   & $1.0629$                   \\
            \hline
            Step 10~(Section~\ref{step-10}) & $[w_5+4\delta_5, w_5+w_4+6\delta_5]$                   & $1.0584$                   \\
            \hline
            Step 11~(Section~\ref{step-11}) & \makecell[c]{$[w_5+4\delta_5, w_5+w_3+6\delta_5]$\\$[w_5+4\delta_5, w_5+7\delta_5 + \sigma]$}              & \makecell[l]{$1.0629$\\$1.0629$}                   \\
            \hline
            Step 12~(Section~\ref{step-12}) & \makecell[c]{ $[w_5+4\delta_5, w_5+4\delta_5+2w_3]$\\$[w_5+4\delta_5, w_5+w_3+5\delta_5 + \sigma]$}          & \makecell[l]{$1.0636$\\$1.0635$}                        \\
            \hline
            Step 13~(Section~\ref{step-13}) & $O^*((1.3279^{1/w_5})^\mu)$                            & $1.0584$                   \\
            \hline
            Step 14~(Section~\ref{step-14}) & $[w_4+2w_3, w_4+6\delta_4]$                            & $1.0638$~*                   \\
            \hline
            Step 15~(Section~\ref{step-15}) & $[w_4+2\delta_4, w_4+6\delta_4]$                       & $1.0638$~*                   \\
            \hline
            Step 16~(Section~\ref{step-16}) & $O^*((1.1279^{1/w_3})^\mu)$                            & $1.0638$~*                   \\
            \hline
        \end{tabular}
    \end{center}
\end{table}

From Table \ref{tb-bv}, we can see that we have four bottlenecks (marked by *): Steps~7, 14, 15, and 16.
In fact, Steps~14, 15, and 16 have the same branching vector $[4w_3, 8w_3]$ under the assumption that $w_4=2w_3$ (for Step~14, the worst branching vector in~\cite{DBLP:conf/sat/Wahlstrom05} is $[4, 8]$).
The branching factor for these three steps will decrease if the value of $w_3$ increases.
On the other hand, the branching factor for Step~7 will decrease if the value of $w_3$ decreases.
We set the best value of $w_3$ to balance them.
If we can either improve Step~7 or improve Steps~9, 10, and 11 together, then we may get a further improvement.
However, the improvement is very limited, and several other bottlenecks will appear.

\section{Concluding Remarks}
In this paper, we show that the SAT problem can be solved in $O^*(1.0638^L)$ time, improving the previous bound in terms of the input length obtained more than ten years ago.
Nowadays, improvement becomes harder and harder. However, SAT is one of the most important problems in exact and parameterized algorithms, and the state-of-the-art algorithms are frequently mentioned in the literature.
For the techniques, although our algorithm, as well as most previous algorithms, is based on case analyses, we introduce a general analysis framework to get a neat and clear analysis. This framework can even be used to simplify the analysis for other similar algorithms based on the measure-and-conquer method.

\section*{Acknowledgements}
This work was supported by the National Natural Science Foundation of China (Grant No. 61972070).
An initial version of this paper was presented at the 24th international conference on theory and applications of Satisfiability testing (SAT 2021)~\cite{DBLP:conf/sat/PengX21}.

\bibliography{SATL}

\end{document}